\documentclass{article}

\usepackage{mathtools,bbm,graphicx,array,enumitem,booktabs,ifthen}

\usepackage[slantedGreeks]{kpfonts}
\usepackage{BOONDOX-cal}

\usepackage{ifart,geometry}

\usepackage{xcolor}
\usepackage[colorlinks=true]{hyperref}
\definecolor{citecolor}{rgb}{0.1,0,0.5}
\definecolor{linkcolor}{rgb}{0,0.1,0.5}
\hypersetup{
	citecolor=citecolor,
	linkcolor=linkcolor
}

\newcommand{\pn}{p}						
\DeclareMathOperator{\prn}{\vdash}		
\newcommand{\cpsum}[2]{Q_{#1,#2}}		
\newcommand{\gpsum}[2]{Q_{#2}(#1)}		
\newcommand{\Ham}{{\mathcal H}}			
\newcommand{\Mon}{{\mathcal M}}			

\newcommand{\dto}{\searrow}								
\newcommand{\upto}{\nearrow}							
\newcommand{\Gfun}{{\otherGamma}}						

\newcommand{\vv}[1]{{\boldsymbol{#1}}}
\newcommand{\N}{{\mathbb N}}
\newcommand{\Z}{{\mathbb Z}}
\newcommand{\R}{{\mathbb R}}
\DeclareMathOperator*{\lO}{{\scriptsize{\mathcal o}}}	

\DeclareMathOperator*{\ass}{\coloneqq}					
\newcommand{\md}{\mathop{\mathrm d\!}}
\DeclareMathOperator{\me}{{\mathop e}}
\DeclareMathOperator{\E}{{\mathbb E}}
\DeclareMathOperator{\V}{{\mathbb V}}
\DeclareMathOperator{\Prob}{{\mathbb P}}
\newcommand{\set}[1]{{\mathcal{#1}}}
\DeclareMathOperator{\const}{{const}}

\DeclareMathOperator{\id}{{\mathbbm 1}}

\newcommand*{\pprime}{{\prime\hspace{-0.75pt}\prime}}

\newcommand*{\mybull}{$\bullet$\quad}		    

\newcommand*{\mysquig}{\raisebox{-0.9ex}{$\widetilde{\rule{5em}{0pt}}$}}				
\newcommand*{\spar}[1]{\par\begin{center}{\mysquig\sc #1\mysquig}\end{center}\par}		

\newtheoremstyle{thmsmall}{2ex}{1ex}{}{}{\sc\bfseries}{}{1em}{}
\newtheoremstyle{thmsmallest}{2ex}{1ex}{}{}{\bfseries}{}{1em}{}
\theoremstyle{thmsmall}
\newtheorem{assum}{Assumption}

\theoremstyle{thmsmallest}

\begin{document}

\title{Limit shapes for Gibbs partitions of sets}

\author{\parbox{2in}{Ibrahim Fatkullin\\{\em University of Arizona}}\parbox{2in}{Jianfei Xue\\{\em University of Missouri}}}

\maketitle

\begin{abstract}
	This study extends a prior investigation of limit shapes for grand canonical Gibbs ensembles of partitions of integers, which was based on analysis of sums of geometric random variables. Here we compute limit shapes for partitions of sets, which lead to the sums of Poisson random variables. Under mild monotonicity assumptions on the energy function, we derive all possible limit shapes arising from different asymptotic behaviors of the energy, and also compute local limit shape profiles for cases in which the limit shape is a step function.
\end{abstract}

\section{Introduction} 
%

In this work we extend the methods introduced in \cite{FatSlas2018} to analyze limit shapes for Gibbs ensembles of partitions of integers to partitions of sets. The primary difference is that the former lead to sums of geometric random variables, while the latter lead to sums of Poisson random variables. Analysis carried out here applies also in a variety of contexts related to measures on permutations and in other combinatorial contexts \cite{erlihson2008limit,betz2011random,ercolani2014cycle,cipriani2013limit,betz2020random,robles2020random}.

We use the same notation and definitions as in \cite{FatSlas2018}, which may be consulted for additional information. A sequence of non-negative integers, $\vv\pn=(\pn_k)$, $k\in\N$, corresponds to a partition of an integer
\begin{equation}\label{eq:def_mass}
	M\,=\,\Mon(\vv\pn)\,\coloneqq\,\sum_{k=1}^{\infty}k\pn_k.
\end{equation}
We say that $\vv\pn$ partitions $M$ (provided the sum above is finite) and denote this by $\vv\pn\prn M$. The function $\Mon(\vv\pn)$ is called the {\em mass} of the partition $\vv\pn$. For a given partition, $\pn_k$-s represent the numbers of summands of size $k$. For example, the partition $21=1+2+2+2+4+4+6$ corresponds to \mbox{$\vv\pn=(1,3,0,2,0,1,0\ldots)$}; respectively, $\Mon(\vv\pn)=21$. We casually employ the polymer physics language referring to summands in a partition as {\em polymers} and to the individual units as {\em monomers}. Thus the mass of a partition may also be thought of as the total number of monomers in the corresponding polymeric system.

\subsection{Gibbs ensembles of set partitions}
Partitions of sets and integers are related: a sequence $\vv\pn=(\pn_k)$ corresponds to a unique partition of the integer $M=\Mon(\vv\pn)$ and also to
\begin{equation}\label{eq:weight}
	\frac{M!}{\prod_{k=1}^\infty(k!)^{\pn_k}\pn_k!}
\end{equation}
different partitions of a set with cardinality $M$. Partitions of integers correspond to {\em indistinguishable} monomers, while partitions of set --- to {\em distinguishable} monomers. The combinatorial factor in formula \eqref{eq:weight} is exactly the number of ways to distribute $M$ labels among unlabeled units in a partition of $M$. It appears as a background weight when we introduce measures on partitions of sets. See, e.g., the treatise by J.~Pitman \cite{pitman1875combinatorial} for a discussion on related topics.

Gibbs measures on partitions are characterized by the weights, $\me^{-\beta\Ham(\vv\pn)},$ where $\Ham(\vv\pn)$ is the energy (Hamiltonian) of a partition $\vv\pn$, and $\beta$ is the inverse temperature of the system. We consider energies of the form,
\begin{equation}\label{eq:hamiltonian}
	\Ham(\vv\pn)
	\,=\,
	\sum_{k=1}^{\infty}E_k \pn_k,
\end{equation}
where the numbers $E_k$ represent the internal energies of polymers of size $k$, $k\in\N$. Thus the total energy of our system is equal to the sum of the individual energies of all polymers.
\paragraph{Canonical Gibbs measures} are defined on partitions of sets with prescribed cardinalities, or alternatively, on sequences $\vv\pn$, such that $\Mon(\vv\pn)=M$. These measures are prescribed by the following probabilities:
\begin{equation}\label{eq:canon_sets}
	\Prob_{M,\beta}\{\vv \pn=\vv P\}
	\,=\,
	\frac{1}{\cpsum{M}{\beta}}\me^{-\beta\Ham(\vv P)}\,
	\prod_{k=1}^\infty\frac{1}{(k!)^{P_k}P_k!};\qquad
	\cpsum{M}{\beta}
	\,=\,
	\sum_{\vv\pn\prn M}\prod_{k=1}^\infty\frac{\me^{-\beta\Ham(\vv P)}}{(k!)^{P_k}P_k!}.
\end{equation}
\paragraph{The grand canonical Gibbs measures} are defined on all partitions with finite mass ($\Mon(\vv\pn)<\infty$) by prescribing
\begin{equation}\label{eq:gcanon_sets}
	\Prob_{\mu,\beta}\{\vv \pn=\vv P\}
	\,=\,
	\frac{1}{\gpsum{\mu}{\beta}}\me^{-\beta\Ham(\vv P)-\mu\Mon(\vv P)}\,
	\prod_{k=1}^\infty\frac{1}{(k!)^{P_k}P_k!};\qquad
	\gpsum{\mu}{\beta}
	\,=\,
	\sum_{M=0}^{\infty}\cpsum{M}{\beta}\,\me^{-\mu M}.
\end{equation}
These measures are superpositions of the canonical measures $\Prob_{M,\beta}$ with weights proportional to \smash{$\me^{-\mu M}$.} The parameter $\mu$ is called {\em chemical potential}. It regulates the expected total mass of the system, \smash{$\E_{\mu,\beta}\Mon(\vv\pn)$}.

As noted by A.~Vershik \cite{vershik1996statistical}, the grand canonical measures are multiplicative:
\begin{equation}\label{eq:gibbs_bell_stat}
	\Prob_{\mu,\beta}\{\vv\pn=\vv P\}
	\,=\,
	\prod_{k=1}^\infty\Prob_{\mu,\beta}^{(k)}\{p_k=P_k\};\qquad
	\Prob_{\mu,\beta}^{(k)}\{p_k=N\}\,=\,\me^{-\alpha_k}\frac{\alpha_k^{N}}{N!},
	\qquad
	\alpha_k\,=\,
	\frac{\me^{-\beta E_k-\mu k}}{k!}.
\end{equation}
This implies that $\pn_k$-s are independent Poisson random variables with parameters $\alpha_k$.  This is the principal difference between the partitions of sets and integers: the latter induce measures for which $\pn_k$-s are geometric random variables, cf. \cite{FatSlas2018}. Note that in some cases (see Section~\ref{sec:noshape}) these multiplicative measures may be supported on sequences for which $\Mon(\vv\pn)$ is infinite, such sequences do not correspond to partitions of any finite integers or sets.

\paragraph{Equivalence of ensembles and the thermodynamic limit.} One of the central problems in statistical mechanics of partitions is to understand various features of the canonical measures in the limit when the mass of the system, $M$, tends to infinity. The grand canonical measures were originally introduced as an aid in this task. Their utility is due to the fact that $\pn_k$-s (which are not independent in the canonical setting) become independent in the grand canonical setting. Equation \mbox{$M=\E_{\mu,\beta}\Mon$} provides the correspondence between $M$ (canonical ensembles) and $\mu$ (grand canonical ensembles). Consequently, instead of the limit as $M\to\infty$, one considers the so-called {\em thermodynamic limit,} $\mu\dto\mu_*$, where $\mu_*$ delimits the interval of convergence of the series $\sum_{k=1}^\infty k\alpha_k$ representing the expected number of monomers in the system. Notice, that as implied by formula \eqref{eq:gibbs_bell_stat}, $\alpha_k$-s are decreasing functions of $\mu$, thus the limit $\mu\to\mu_*$ is always taken from above; see Section~\ref{ssec:informal} for more details regarding the value of $\mu_*$.

Yu.~Yakubovich~\cite{yakubovich1995asymptotics} and A.~Vershik~\cite{vershik1996statistical} established that as far as the distributions of the (appropriately rescaled) $\pn_k$-s are concerned, the canonical and grand canonical measures are equivalent in respective limits as $M\to\infty$ and $\mu\dto\mu_*$, provided the latter has a limit shape (see below).
In this work we only study the grand canonical measures and whenever we omit the subscripts $\mu$ and $\beta$, we imply that the quantities in question are computed with respect to a grand canonical measure with corresponding parameters. Results regarding canonical measures may then be deduced whenever the equivalence of ensembles holds.

\subsection{Limit shapes}\label{ssec:lim_shapes}
Following A.~Vershik \cite{vershik1996statistical}, we define the {\em size distribution function} of a partition as
\begin{equation}\label{eq:SDF_Ver}
	f(x;\vv\pn)\,\coloneqq\,\sum_{k\,\geq\,x}\pn_k;
	\qquad x\in \R^+\,\coloneqq\, [0,\infty).
\end{equation}
The graph of this function is the boundary of a Young (Ferrers) diagram corresponding to the partition $\vv\pn$. The rescaled size distribution function is defined as
\begin{equation}\label{eq:scaled_dist_fun}
	F_\mu(x;\vv\pn)\,\coloneqq\,\frac{\varkappa}{\E\Mon}\;f(\varkappa x;\vv\pn) = \frac{\varkappa}{\E\Mon}\;  \sum_{k\,\geq\,\varkappa x}\pn_k;\qquad x\in \R^+.
\end{equation}
The factor $\E\Mon$ is chosen to make $\E F_\mu(x;\vv\pn)$ integrate to one; the parameter $\varkappa=\varkappa(\mu)$ controls the scaling in the ``horizontal'' direction. Scaling in the ``vertical'' direction is then automatically determined as $\E\Mon/\varkappa$. The dependence $\varkappa(\mu)$ must be carefully chosen to obtain a sensible behavior for $F_\mu(x;\vv\pn)$ in the thermodynamic limit. Whenever, with a proper choice of scaling, $F_\mu(x;\vv\pn)$ converges (in an appropriate sense) to some nonzero deterministic function,
\begin{equation}
	F(x)\,\ass\,\lim_{\mu\dto\mu_*}F_\mu(x;\vv\pn),
\end{equation}
the latter is called {\em limit shape} corresponding to the energies $E_k$. Its graph is the limiting curve (in probability) for boundaries of the rescaled Young diagrams for respective Gibbs partitions.

As $\pn_k$ are independent Poisson random variables with parameters $\alpha_k$, see formula \eqref{eq:gibbs_bell_stat}, their size distribution functions are also Poisson-distributed, with parameters equal to  the sums of the corresponding $\alpha_k$-s. Therefore, such quantities as the expected mass of the partition, expectation and variance of the rescaled size distribution function, are given by the following sums:
\begin{equation}\label{eq:poisson_sum}
	\E\Mon\,=\,\sum_{k=1}^\infty k\alpha_k;\qquad
	\E F_\mu(x;\vv\pn)
	\,=\,\frac{\varkappa}{\E\Mon}
	\sum_{k\,\geq\,\varkappa x}\alpha_k;\qquad
	\V F_\mu(x;\vv\pn)
	\,=\,\frac{\varkappa}{\E\Mon}\E F_\mu(x;\vv\pn).
\end{equation}
Thus this work is essentially a study of asymptotic behaviors of such sums as $\mu$ tends to $\mu_*$ --- the boundary of their interval of convergence, see equation \eqref{eq:mu_star} below.

\subsection{Known results regarding limit shapes for partitions of sets}

Most of the classical limit shape studies are motivated by problems in representation theory of the symmetric group, or by various combinatorial constructions; see e.g., \cite{bogachev2015unified}. In particular, one result relevant to our work concerns the measure induced on partitions by the uniform (Dirichlet-Haar) measure on the symmetric group via its cycle structure. This measure was analyzed in great detail by J.~F.~C.~Kingman \cite{kingman1977population}, and A.~Vershik and A.~Shmidt \cite{vershik1977limit, vershik1978limit}. It corresponds to setting \smash{$\alpha_k\,=\,\me^{-\mu k}\!/k$} in formula \eqref{eq:gibbs_bell_stat}. There is no limit shape in this case. Instead, the distributions of cycle lengths converge (after appropriate transformation) to a Poisson-Dirichlet or a closely related Griffiths-Engen-McCloskey (GEM) distribution \cite{kingman1975random,holst2001poisson,arratia2006tale}.

Ensembles represented by formula \eqref{eq:gibbs_bell_stat} may also be obtained by assigning length-dependent weights to the cycles of a permutation (rather than to the members of a partition). V.~Betz, D.~Ueltschi, and Y.~Velenik \cite{betz2011random}; or N.~Ercolani and D.~Ueltschi \cite{ercolani2014cycle} computed various statistical quantities, such as the expected cycle length, or total number of (finite) cycles for several specific ensembles in this context.

Another result, obtained by Yu.~Yakubovich \cite{yakubovich1995asymptotics}, concerns the uniform measure on partitions of sets. This corresponds to \smash{$\alpha_k\,=\,\me^{-\mu k}\!/k!$} or $\beta=0$ in our terminology. In this case the limit shape is given by the step function, $\id_{[0,1]}(x)$. The appropriate scaling is $\varkappa(\mu)=\me^{-\mu}$.

The so-called {\em expansive} case, corresponding in our notation to \smash{$\alpha_k\,\sim\,k^{\,d-1}\me^{-\mu k}$} with $d>0$ was addressed directly in the context of canonical Gibbs ensembles by M.~Erlihson and B.~Granovsky \cite{erlihson2008limit} and also generalized  by A.~Cipriani and D.~Zeindler~\cite{cipriani2013limit}. In this case the thermodynamic limit is attained as $\mu\dto0$ and the appropriate scaling is $\varkappa=1/\mu$. The limit shape is given by the incomplete gamma function, $\Gfun(x;d)$.

In this work we fill the gaps in existing results and establish a complete classification of possible asymptotic behaviors of the scaled size distribution functions for all possible asymptotic behaviors of the energies $E_k$ under mild monotonicity assumptions. In particular, we determine for which $E_k$-s one gets a limit process rather than a limit shape, when the limit shape is given by the step function, or by the incomplete gamma function without any assumptions regarding specific functional form of the asymptotic behavior of $E_k$-s (or $\alpha_k$-s). For scenarios in which the limit shape is given by the step function as in formula \eqref{eq:step_fun_shape}, we also study its local profile in the vicinity of the discontinuity point. Our methods are based on a mixture of direct probabilistic and analytical techniques similar to those utilized in \cite{chatterjee2014fluctuations,FatSlas2018} and are somewhat simpler and more elementary than the methods based on saddle-point-type techniques for generating functions, as in the aforementioned studies.

%
%
%
%
%
%
%
%

\subsection{Informal statement of main results}\label{ssec:informal}
In this section we summarize our main results in an informal manner.  As formula \eqref{eq:gibbs_bell_stat} suggests, the $k!$ factor may be absorbed from the denominator of $\alpha_k$-s into the exponential, and thus we introduce
\begin{equation}\label{eq:smallE}
	u(k)\,\coloneqq\,
	\beta E_k\,+\,\ln\Gfun(k+1).
\end{equation}
It is more straightforward to work with $u(k)$ rather than $E_k$-s, and we employ it in most of our proofs and calculations. We assume that the energies may be extended from $k\in\N$ into \smash{$\R^+$} in a smooth and monotone manner:

\vspace{1ex}
\begin{assum} \label{Standing Assump}
	There exists a twice differentiable function $u:\R^+\to \R$ satisfying equation \eqref{eq:smallE}, such that 
	$\lim_{x\to \infty} x^2 u^\pprime(x)$ exists in $\R\cup\{-\infty,+\infty\}$.
\end{assum}
\vspace{2ex}
%

This assumption is more technical than conceptual. The differentiability allows us to use the mean value theorem at various points in the proofs; it may always be achieved, as any function defined on integers can always be smoothly extended to reals. The condition on $u^\pprime(x)$ allows us to avoid dealing with $\lim$-sup/infs at various stages in the proofs. It is not entirely essential, however: if it fails, there may exist subsequences $\{k_i\}$ along which the asymptotic behaviors of \smash{$\pn_{k_i}$\!-s} differ. As all $\pn_k$-s are independent in the grand canonical ensembles, one could analyze  all such subsequences separately using the same methods as presented here.

The thermodynamic limit as attained as $\mu$ approaches $\mu_*$, a value such that \smash{$\lim_{\mu\dto\mu_*}\E\Mon = \infty$}. Employing formula \eqref{eq:smallE}, we can express the expected number of monomers as
\begin{equation}\label{eq:sum_part_expl}
	\E\Mon\,=\,\sum_{k=1}^\infty k\me^{-\mu k - u(k)}.
\end{equation}
Using the root test, we can conclude that this sum converges for all $\mu>\mu_*$ and diverges for all $\mu<\mu_*$, where
%
\begin{equation}\label{eq:mu_star}
	\mu_*\,\ass\,-\lim_{x\to\infty}\frac{u(x)}{x}
	\,=\,-\lim_{x\to\infty}u^\prime(x).
\end{equation}
In particular, if $\mu_*=-\infty$, $\E\Mon$ is finite for all $\mu$, while if $\mu_*=\infty$, $\E\Mon$ is always infinite. Notice that Assumption \ref{Standing Assump} allows us to use a regular limit instead of lim-inf in equation~\eqref{eq:mu_star}, while Lemmas~\ref{lem: structure of u, const limit} and \ref{lem: Assm2impliesAssm1} imply the second equality.

As discussed earlier in Section~\ref{ssec:lim_shapes}, the limit shapes appear provided that the sizes of Young diagram cells are appropriately rescaled as $\mu\dto\mu_*$. In most cases, the required choice of the horizontal scaling factor is $\varkappa = \varkappa(\mu)$, the largest solution of
\begin{equation}\label{eq:step_scaling}
	u^\prime(\varkappa)\,=\,-\mu.
\end{equation}
One can see from the formula \eqref{eq:sum_part_expl}, that this value of $k$ corresponds to the maximum of the expression inside of the exponent, i.e., to the range of the most likely polymer sizes. In scenarios when $u^\pprime(x)>0$ for all large enough $x$, e.g. when $x\in(X,\infty)$ for some value $X$, this implies that as $\mu\dto\mu_*$, equation \eqref{eq:step_scaling} has a unique solution in $(X,\infty)$ and thus its largest solution is well-defined. In the other cases the appropriate scaling is $\varkappa=1/(\mu-\mu_*)$ see below for details.


\newpage
%
\spar{Limit shape scenarios}
\vspace{1ex}

\paragraph{\label{par:sub}\sc\bfseries Subcritical regime,\quad $\boldsymbol{\lim_{x\to\infty} x^2 u^\pprime(x) = -\infty .}$} 
The concept of limit shape is not applicable here; there are two possible situations:
\begin{enumerate}[label=\bfseries\alph*)]
	\item $\boldsymbol{\mu_*\,=\,\infty.}$ In such systems the multiplicative measures defined by formula \eqref{eq:gibbs_bell_stat} are supported on sequences, $\vv\pn$, for which $\Mon(\vv\pn)=\infty$, i.e., they do not correspond to partitions of finite sets or integers. Consequently, the grand canonical probability measures cannot be defined via formula \eqref{eq:gcanon_sets}, while the sums representing the size distribution functions \eqref{eq:SDF_Ver} diverge almost surely. One example when this scenario occurs is whenever $u(x)$ tends to negative infinity as fast or faster than $-x(\ln x)^p$, $p>0$.
	\item $\boldsymbol{\mu_*\,\in\,\R.}$ In such systems the expected number of monomers remains finite for all values of $\mu$, i.e., $\E\Mon<\infty$ when $\mu=\mu_*$. This implies that the thermodynamic limit cannot be achieved in the setting of grand canonical ensembles and the size distribution functions remain discrete as $\mu\dto\mu_*$. One family of ensembles with this behavior is prescibed by $u(x)=(\ln x)^p$, $p>1$.
\end{enumerate}

\vspace{1ex}
\paragraph{\label{par:super}\sc\bfseries Supercritical regime,\quad  $\boldsymbol{\lim_{x\to\infty} x^2 u^\pprime(x) = \infty .}$}
In this regime most polymers have sizes in the vicinity of $k\sim \varkappa(\mu)$ as prescribed in equation \eqref{eq:step_scaling}. The limit shape is attained under the scaling of $\varkappa(\mu)$ and is given by the step function,
\begin{equation}\label{eq:step_fun_shape}
	F(x)\,=\,
	\id_{[0,1]}(x).
\end{equation}
This regime encompasses all cases which yield the same limit shape as the uniform measure on partitions of sets for which $u(x)=\ln\Gfun(x+1)$, cf. \cite{yakubovich1995asymptotics}. Other examples in this category include $u(x)=x^p$, $p>1$; $u(x)=x(\ln x)^p$, $p>0$, etc.

\vspace{2ex}
\paragraph{\label{par:crit}\sc\bfseries Critical regime,\quad  $\boldsymbol{\lim_{x\to\infty} x^2 u^\pprime(x) \in \R .}$}
In this regime the function $u(x)$ must have the form (Lemma~\ref{lem: structure of u, const limit})
\begin{equation}
	u(x) \,=\, -\mu_* x\,+\,(1-d)\ln x \,+\, v(x),\quad\text{where}\quad d\in \R,\quad v(x) \ll \ln x.
\end{equation}
The value of $\mu_*$ in this case is necessarily finite and without loss of generality may be assumed to be equal to zero. The limit shape behavior is determined by the finer (sublinear) features of $u(x)$. Let us relate these features back to the energies $E_k$ that were used to introduce Gibbs ensembles in the first place. Using Stirling's approximation, $\ln\Gfun(k+1)\sim k\ln k - k + 1/2\,\ln k +\cdots$ in formula \eqref{eq:smallE}, we can see that unless $\beta E_k\sim -k\ln k$, the system is always (for all $\beta$) either in the sub or supercritical regime. The critical regime is only possible when the leading order asymptotics of $\beta E_k$ and $\ln\Gfun(k+1)$ cancel each other exactly, in which case the interplay between the lower order behavior of $\beta E_k$-s and $1/2\,\ln k$ from $\ln\Gfun(k+1)$ determines the outcome, i.e., the value of $d$ and the asymptotics of $v(x)$ as $x\to\infty$. Different scenarios arise are possible:
\begin{enumerate}[label=\bfseries\alph*)]
	\item {\bfseries Either $\boldsymbol{d<0}$,\quad or $\boldsymbol{d=0}$~~and~~$\boldsymbol{v(x)\to\infty}$.}\quad If $d<-1$, $\E\Mon$ remains finite for all values of $\mu$, as in the subcritical regime (b) above. In all other cases, the rescaled size distribution functions may only converge to a trivial function identically equal to zero. This regime generalizes the {\em convergent} case in \cite{erlihson2008limit} where the conditions $d<0$ and $v(x)\sim\const$ are assumed.


	\item{\bfseries Either $\boldsymbol{d>0}$,\quad or $\boldsymbol{d=0}$~~and~~$\boldsymbol{v(x)\to-\infty}$.}\quad The limit shape is given by the (rescaled) incomplete gamma function,
	      \begin{equation}\label{eq:GammaFun}
		      F(x)
		      \,=\,
		      \frac{\Gfun(x;d)}{\Gfun(d+1)};\qquad
		      \Gfun(x;d)\;\ass\;\int_x^\infty y^{d-1}\me^{-y}\md y.
	      \end{equation}
	      It is attained with the scaling $\varkappa=1/\mu$, asymptotically equivalent to \eqref{eq:step_scaling}. This regime generalizes the {\em expansive} case from \cite{erlihson2008limit} where $d>0$ and $v(x)\sim\const$ are assumed; and $d>1$, $v(x)=-j\ln\ln x$ case from \cite{cipriani2013limit}. For example, our result covers cases such as $d=0$, $v(x)=-(\ln x)^p$, $p\in(0,1)$; or \mbox{$v(x)\sim-\ln\cdots\ln x$}.

	\item{\bfseries$\boldsymbol{d=0}$~and~$\boldsymbol{v(x)\to C}$,}\quad In this case the scaled size distribution functions (under the same scaling, \mbox{$\varkappa=1/\mu$}) converge in distribution to a Poisson process which starts at $0$ when \mbox{$x=\infty$} and whose jumps are distributed with density $\me^{-C-x}/x$. As $x\dto0$, this process tends to infinity almost surely. It is closely related to the Poisson-Dirichlet distribution \cite{kingman1975random,holst2001poisson} which describes the distribution of the (scaled) cycle lengths in a random permutation; see Section~\ref{sec:limitProcess}.
\end{enumerate}

\vspace{2ex}
\spar{Local profiles for the step function scenarios}
\vspace{1ex}
\noindent As discussed above, in the supercritical regime the limit shape is given by the step function. It is, therefore, natural to study the local behavior of the size distribution functions near the point of discontinuity. In order to do that, we consider the shifted size distribution function and its rescaled version, cf. equations \eqref{eq:SDF_Ver} and \eqref{eq:scaled_dist_fun}:
\begin{equation}\label{eq:shift_shape}
	g(x;\vv\pn)
	\,\coloneqq\,
	\sum_{k\,\geq\,x + \varkappa}\pn_k
	\qquad \text{and} \qquad
	G_\mu(x;\vv\pn)
	\,\coloneqq\,
	\dfrac{\varkappa}{\E \Mon}\,g(\zeta x;\vv\pn),
	\quad x\,\in\,\R.
\end{equation}
Notice that the shift is given by $\varkappa(\mu)$, solution of equation \eqref{eq:step_scaling}; the ``vertical'' scaling factor, $\E \Mon/\varkappa$, remains as in equation \eqref{eq:scaled_dist_fun}, while the ``horizontal'' scaling factor, $\zeta=\zeta(\mu)$, is different and must be determined from further analysis, see Section~\ref{sec:locProf}. Let
\begin{equation}
	G(x)\,\ass\,\lim_{\mu\dto\mu_*} G_\mu(x;\vv\pn)
\end{equation}
be the local limit shape. The following three regimes are possible:
\vspace{1ex}
\paragraph{\sc\bfseries Gaussian regime,\quad$\boldsymbol{\lim_{x\to \infty}u^\pprime (x) = 0}$.} The proper scaling is prescribed by \smash{$\zeta = 1 / \sqrt{u^\pprime (\varkappa)}$}, and the local shape is given by the Gaussian integral,
\begin{equation}
	G(x) \,=\,\frac{1}{\sqrt{2\pi}}\int_x^\infty  \me^{-y^2/2}\md y.
\end{equation}

\paragraph{\sc\bfseries Discrete Gaussian regime,\quad$\boldsymbol{\lim_{x\to \infty}u^\pprime (x)=C \in(0,\infty)}$.} No additional scaling is needed, i.e., $\zeta =1$, and the local shape is a discrete Gaussian distribution,
\begin{equation}
	G(x) =\dfrac 1{Q} \sum_{k\,\geq\,x} \me^{-Ck^2/2},\quad \text{where}\quad Q = \sum_{k\,\in\,\Z} \me^{-C k^2/2}.
\end{equation}

\paragraph{\sc\bfseries Hard step function regime\quad$\boldsymbol{\lim_{x\to \infty}u^\pprime (x) = \infty}$.} In this regime the size distribution function asymptotically concentrates on finitely many values of $k$, and the local limit shape remains a step function under any rescaling such that $\zeta\to\infty$ as $\mu\dto\mu_*$.

%

\section{Limit shape theorems}
In this section we present the formal statements and calculations regarding the limit shapes for the grand canonical Gibbs ensembles of set partitions. The Assumption~\ref{Standing Assump} on the function $u(\cdot)$ introduced in section \ref{ssec:informal} is enacted. As can be expected from Section~\ref{ssec:lim_shapes}, the calculations that we carry out in order to establish various limit shape statements involve various sums as in equation \eqref{eq:poisson_sum}. Let us introduce some additional notation for these sums:
\begin{equation}\label{eq:main_sum}
	S_\mu(a,b)\,\ass\,\E\sum_{a\,\leq\,k\,<\,b}\pn_k
	\,=\,\sum_{a\,\leq\,k\,<\,b}\me^{-\mu k-u(k)};\qquad
	S_\mu(a)\ass S_\mu(a,\infty);\qquad
	S_\mu\ass S_\mu(1,\infty).
\end{equation}
\subsection{The step function regime}\label{ssec:StepFun}
We start from the cases where the limit shape is the step function $F(x)=\id_{[0,1]}(x)$. This scenario occurs when \smash{$\lim_{x\to \infty} x^2 u^\pprime(x)=\infty$}.
In this regime the sums in \eqref{eq:main_sum} are dominated by the terms near the maximum of $-\mu k - u(k)$, i.e.,  in the vicinity of $k=\varkappa(\mu)$ determined by equation \eqref{eq:step_scaling}. Notice also that
\begin{equation}\label{eq:sumMon}
	\E\Mon\,=\,\sum_{k=1}^\infty k\me^{-\mu k - u(k)}\,=\, \sum_{k=1}^\infty \me^{-\mu k - \left[u(k)-\ln k\right]}.
\end{equation}
Therefore, by Lemma \ref{lem: Assm2impliesAssm1}, we have $\lim_{\mu\dto\mu_*} \E \Mon\,=\,\infty$, and the thermodynamic limit is achieved as  $\mu\dto\mu_*$.

\begin{theorem}\label{thn:step_1}
	Suppose $u(\cdot)$ satisfies $\lim_{x\to \infty} x^2 u^\pprime (x) = \infty$. Let $F(x)=\id_{[0,1]}(x)$, and $\mu$ and $\varkappa$ be related as prescribed in equation \eqref{eq:step_scaling}. Then for each $\lambda_1<1$, $\lambda_2>1$, and $\epsilon>0$,
	\begin{equation}
		\lim_{\mu\dto\mu_*}\Prob \left\{\sup_{x\in\R^+\setminus(\lambda_1,\lambda_2)}\big|F_\mu (x;\vv\pn) - F(x)\big| > \epsilon\right\}\,=\,0.
	\end{equation}
\end{theorem}

\begin{proof}
	As $F_\mu(x;\vv\pn)$ is decreasing and $F(x)=\id_{[0,1]}(x)$, it is enough to show that for all $x\in[0,1)\cup (1,\infty)$,
	\begin{equation}
		\lim_{\mu\dto\mu_*}\Prob \left\{|F_\mu (x;\vv\pn) - F(x)| > \epsilon\right\} = 0.
	\end{equation}
	Using the triangle inequality, we can split this into two parts:
	\begin{equation} \label{eqn: step fcn two convergence}
		\lim_{\mu\dto\mu_*}\Prob \left\{|F_\mu (x;\vv\pn) - \E F_\mu(x;\vv\pn) | > \epsilon\right\} = 0; \qquad
		\lim_{\mu\dto\mu_*} \E F_\mu(x;\vv\pn) = F(x).
	\end{equation}
	For the first limit, by Chebyshev's inequality, we get
	\begin{equation}\label{eq:th1e1}
		\Prob\left\{|F_\mu (x;\vv\pn) - \E F_\mu(x;\vv\pn) | > \epsilon\right\}
		\leq
		\dfrac{1}{\epsilon^{2}} \V F_\mu(x;\vv\pn)
		\,=\,
		\dfrac{1}{\epsilon^{2}}  \frac{\varkappa}{\E\Mon}\E F_\mu(x;\vv\pn).
	\end{equation}
	By Lemma \ref{lem: polymer and monomer number}, $\varkappa/\E\Mon\sim 1/S_\mu$, while by Lemma \ref{lem: Assm2impliesAssm1}, $\lim_{\mu\dto \mu_*}S_\mu = \infty$. Thus the right-hand side of equation~\eqref{eq:th1e1} tends to 0 as $\mu\dto\mu_*$, whenever the second limit in equation \eqref{eqn: step fcn two convergence} holds as well.

	To prove the latter, we note that
	\begin{equation}
		\E F_\mu(x;\vv\pn) \,=\, \frac{\varkappa}{\E\Mon} S_\mu(\varkappa x)
		\,=\, \frac{\varkappa S_\mu}{\E\Mon} \frac{S_\mu(\varkappa x)}{S_\mu}.
	\end{equation}
	The same Lemma~\ref{lem: polymer and monomer number} implies that $\varkappa S_\mu\sim\E\Mon$, thus it remains to show that, for all $x\in [0,1)\cup(1,\infty)$, $\lim_{\mu\dto\mu_*} S_\mu(\varkappa x) / S_\mu \,=\,\id_{[0,1]}(x)$. This follows from the fact that in this regime the principal contribution into the sums in equation \eqref{eq:main_sum} comes from the terms near $k=\varkappa(\mu)$, which is proven in Lemma~\ref{lem:step_1} below.
\end{proof}

\begin{lemma}\label{lem:step_1}
	Assume that a twice-differentiable function $u(\cdot)$ satisfies $\lim_{x\to \infty }x^2 u^\pprime(x) = \infty$.
	Fix arbitrary $\lambda_1$ and $\lambda_2$ satisfying $0<\lambda_1<1<\lambda_2$. Let $\varkappa(\mu)$ be the solution of \eqref{eq:step_scaling}.
	Then
	\begin{equation} \label {eqn: step fcn case density}
		\lim_{\mu\dto\mu_*}S_\mu(\varkappa\lambda_1,\varkappa\lambda_2)/S_\mu
		\,=\,1.
	\end{equation}
\end{lemma}

\begin{proof}
	To prove \eqref{eqn: step fcn case density}, it is sufficient to show that $S_\mu(\varkappa\lambda_2)/S_\mu\to0$ and $S_\mu(1,\varkappa\lambda_1)/S_\mu \to0$ as $\mu\dto\mu_*$. We split the proof into the following two steps:

	\smallskip
	{\bf \noindent Step 1.}
	For the first assertion, let $\varepsilon = \lambda_2 -1$; $\varkappa_n = (1+n\varepsilon)\varkappa$.
	Then we have
	\begin{equation}
		S_\mu(\varkappa\lambda_2)
		\,=\,
		\sum_{n=1}^\infty S_\mu(\varkappa_n,\varkappa_{n+1}).
	\end{equation}
	Let $\alpha_n = \me^{-\mu\varkappa_n - u\left(\varkappa_n\right)}$, cf. equation \eqref{eq:gibbs_bell_stat}. By definition of $\varkappa$, the quantity $-\mu k - u(k)$ is decreasing whenever $k\geq\varkappa$, therefore $S_\mu(\varkappa_n,\varkappa_{n+1})\leq(\varkappa_{n+1}-\varkappa_n)\alpha_n =\varepsilon\varkappa\alpha_n$, and
	\begin{equation}
		S_\mu(\varkappa\lambda_2)
		\,\leq\,
		\varepsilon\varkappa
		\sum_{n=1}^\infty\alpha_n.
	\end{equation}
	Observe that $
		{\alpha_{n+1}}/{\alpha_n}
		\,=\,
		\me^{-\mu(\varkappa_{n+1} - \varkappa_n) - \left[u(\varkappa_{n+1}) - u(\varkappa_n)\right]}
		\,=\,
		\me^{-\varepsilon\varkappa[\mu + u^\prime (\xi)]}$
	for some $\xi\in[\varkappa_n,\varkappa_{n+1}]$. As $u^\prime$ is increasing, \smash{$u^\prime(\xi)\geq u^\prime(\varkappa_n)\geq u^\prime(\varkappa_1)$}. Therefore \smash{${\alpha_{n+1}}/{\alpha_n}
		\, \leq\,
		\me^{-\varepsilon\varkappa\left(\mu + u^\prime (\varkappa_1)\right)}$},
	which implies that
	\begin{equation}
		\sum_{n=1}^\infty\alpha_n\,\leq\,
		\alpha_1
		\sum_{n=0}^\infty
		\me^{-n\varepsilon\varkappa[\mu + u^\prime (\varkappa_1)]}
		\,=\,
		\frac{\alpha_1}{1-\,\me^{-\varepsilon\varkappa[\mu + u^\prime (\varkappa_1)]}}.
	\end{equation}
	Let $\tilde\varkappa = (1+\varepsilon/2) \varkappa$ and $\tilde\alpha_0 = \me^{-\mu \tilde\varkappa - u(\tilde\varkappa)}$.
	We estimate $S_\mu(\varkappa,\tilde\varkappa)
		\,\geq\,
		(\tilde\varkappa-\varkappa)\me^{-\mu\tilde\varkappa - u(\tilde\varkappa)}
		\,=\,
		\varepsilon\varkappa\tilde\alpha_0/2$ for terms around $\varkappa$, and also observe that ${\alpha_1}/{\tilde\alpha_0}
		\,\leq\,
		\me^{-[\mu+u^\prime (\tilde\varkappa)] \varepsilon\varkappa/2 }$.
	Then
	\begin{equation}
		S_\mu(\varkappa\lambda_2)/S_\mu(\varkappa,\tilde\varkappa)
		\,\leq\,
		\frac{2\alpha_1}{\tilde\alpha_0}\frac{1}{1-\me^{-\varepsilon\varkappa[\mu + u^\prime (\varkappa_1)]}}
		\,\leq\,
		\frac{\me^{-[\mu+u^\prime (\tilde\varkappa)]\varepsilon\varkappa/2 }}{1-\me^{-\varepsilon\varkappa[\mu + u^\prime (\varkappa_1)]}}.
	\end{equation}
	Recall the relation \eqref{eq:step_scaling} between $\mu$ and $\varkappa$, and that $x^2 u^\pprime(x)\to \infty$ by assumption. This implies that there exists some $\xi\in [\varkappa,\varkappa_1]$, such that
	\begin{equation}
		\me^{-\varepsilon\varkappa[\mu + u^\prime (\varkappa_1)]}
		\,=\,
		\me^{ -u^\pprime(\xi) \varepsilon^2 \varkappa^2 }
		\,\to\, 0,\quad
		\text{as}\quad\mu\dto\mu_*\quad
		\text{(or, equivalently, as $\varkappa\to\infty$)}.
	\end{equation}
	Similarly,
	\begin{equation}
		\me^{-[\mu+u^\prime (\tilde\varkappa)] \varepsilon\varkappa/2 }
		\to 0,\quad
		\text{as}~\mu\dto\mu_*.
	\end{equation}
	Then we conclude,
	\begin{equation} \label {eqn: to the right of kappa limit}
		\lim_{\mu\dto\mu_*} S_\mu(\varkappa\lambda_2)/S_\mu(\varkappa,\tilde\varkappa)
		\,=\,0,
	\end{equation}
	and therefore $S_\mu(\varkappa\lambda_2)/S_\mu\dto0$.

	\smallskip
	{\bf \noindent Step 2.}
	For the second assertion, we first introduce a cut-off for $k$. Because $\lim_{x\to \infty} x^2 u^\pprime(x) = \infty$, we can find $N$ such that $u^\pprime(x) >0$ for $x\geq N$. In other words, $u^\prime(x)$ is increasing on $[N,\infty)$. Let us show that
	\begin{equation} \label {eqn: drop the first N}
		\lim_{\mu\dto \mu_*} \dfrac{S_\mu(1,N)}{S_\mu}
		\,=\,0.
	\end{equation}
	If  $\mu_* = -\infty$,
	let $r = \max_{1\leq x\leq N-1} u^\prime(x)$. Then for any $\mu < - r$, we have $-\mu k - u(k)$ is increasing on $1\leq k \leq \varkappa$.
	As $\varkappa \to\infty$, equation \eqref{eqn: drop the first N} follows. On the other hand, if $\mu_*>-\infty$,
	\begin{equation}
		\lim_{\mu\dto \mu_*}S_\mu(1,N)
		\,=\,
		\sum_{k=1}^{N-1} \me^{-\mu_* k - u(k)}
		\,\leq\,
		(N-1) \max_{1\leq k\leq N-1}  \me^{-\mu_* k - u(k)}.
	\end{equation}
	As $\lim_{\mu\dto \mu_*} S_\mu = \infty$, equation \eqref{eqn: drop the first N} follows in this case as well. Thus in order to show that the ratio $S_\mu(1,\varkappa\lambda_1)/S_\mu$ tends to zero, it suffices to have
	\begin{equation}
		S_\mu(N,\varkappa\lambda_1)/S_\mu \to 0.
	\end{equation}
	To this end, we follow a similar idea as in Step 1.
	Let $\varepsilon = 1 - \lambda_1$ and $L = \min \left\{ n\in \N : N \geq (1- n\varepsilon)\varkappa \right\}.$
	Let $\varkappa_n = (1- n\varepsilon)\varkappa$ for $x=1,\ldots, L-1$ and $\varkappa_L = N$.
	Then we have
	\begin{equation}
		S_\mu(N,\varkappa\lambda_1)
		\,=\,
		\sum_{n=1}^{L-1} S_\mu(\varkappa_{n+1},\varkappa_n).
	\end{equation}
	Let $\alpha_n = \me^{-\mu\varkappa_n - u\left(\varkappa_n\right)}$. When $k\in [N,\varkappa]$, the quantity $-\mu k - u(k)$ is increasing, which implies that \mbox{$S_\mu(\varkappa_{n+1},\varkappa_n)
		\leq
		(\varkappa_n- \varkappa_{n+1})\alpha_n
		\leq
		\varepsilon\varkappa\alpha_n$}.
	Thus
	\begin{equation}
		S_\mu(N,\varkappa\lambda_1)
		\,\leq\,
		\varepsilon\varkappa
		\sum_{n=1}^{L-1}\alpha_n.
	\end{equation}
	Observe that, when $n=1,\ldots,L-2$, $
		{\alpha_{n+1}}/{\alpha_n}
		\,=\,
		\me^{-\mu(\varkappa_{n+1} - \varkappa_n) - \left[u(\varkappa_{n+1}) - u(\varkappa_n)\right]}
		\,=\,
		\me^{\,\varepsilon\varkappa\left(\mu + u^\prime (\xi)\right)}$
	for some $\xi\in[\varkappa_{n+1},\varkappa_n]$. As $u^\prime(x)$ is increasing on $[N,\infty)$, $u^\prime(\xi)\leq u^\prime(\varkappa_n)\leq u^\prime(\varkappa_1)$. $
		{\alpha_{n+1}}/{\alpha_n}
		\, \leq\,
		\me^{\,\varepsilon\varkappa\left(\mu + u^\prime (\varkappa_1)\right)}$.
	This yields that
	\begin{equation}
		S_\mu(N,\varkappa\lambda_1) \,\leq\,
		\varepsilon\varkappa \sum_{n=1}^{L-1}\alpha_n\,\leq\,
		\varepsilon\varkappa \alpha_1
		\sum_{n=0}^\infty
		\me^{\,n\varepsilon\varkappa\left(\mu + u^\prime (\varkappa_1)\right)}
		\,=\,
		\frac{\varepsilon\varkappa \alpha_1}{1-\,\me^{\,\varepsilon\varkappa\left(\mu + u^\prime (\varkappa_1)\right)}}.
	\end{equation}
	Let $\tilde\varkappa = (1-\varepsilon/2) \varkappa$. Using the same argument as for validating formula \eqref{eqn: to the right of kappa limit} above, we conclude that $S_\mu(N,\varkappa\lambda_1)/S_\mu \to 0$.
\end{proof}

\begin{lemma} \label{lem: polymer and monomer number}
	Assume that $u(\cdot)$ and $\varkappa(\mu)$ are as in Lemma~\ref{lem:step_1}. Then
	\begin{equation}
		\lim_{\mu\dto\mu_*}
		\varkappa S_\mu/\E\Mon\,=\,1.
	\end{equation}
\end{lemma}

\begin{proof}
	Recall that $\E\Mon$ may be treated exactly as $S_\mu$ if we absorb the extra factor of $k$ into $u(k)$, see equation \eqref{eq:sumMon}. Thus as in the previous lemma, the principal contributions into $\E\Mon$ come from the terms in the vicinity of $\hat\varkappa$, the largest solution of
	\begin{equation}
		u^\prime(\hat\varkappa) - 1/\hat\varkappa\,=\,-\mu.
	\end{equation}
	Such $\hat\varkappa$ exists, as $u^\prime(\hat\varkappa) - 1/\hat\varkappa$ is increasing for large $\hat\varkappa$. Let us show that $\lim_{\mu\dto\mu_*}(\varkappa/\hat\varkappa)= 1$.
	Indeed, as $\mu = -u^\prime(\varkappa)$,
	\begin{equation}
		u^\prime(\varkappa) \,=\, u^\prime(\hat\varkappa) - 1/\hat\varkappa.
	\end{equation}
	Thus by the mean value theorem, there exists $\xi\in[\varkappa,\hat\varkappa]$ (note that $\hat\varkappa>\varkappa$), such that
	\begin{equation}
		u^\pprime(\xi) (\hat\varkappa - \varkappa) \,=\, 1/\hat\varkappa,
		\text{\quad or equivalently,\quad}
		\xi^2 u^\pprime(\xi)\;\big[\hat\varkappa(\hat\varkappa - \varkappa)/\xi^2\big]\,=\,1.
	\end{equation}
	By assumption, $\lim_{\xi\to\infty}\xi^2 u^\pprime(\xi) =\infty$, and therefore we must have
	\begin{equation}
		\lim_{\mu\dto\mu_*}\big[\hat\varkappa(\hat\varkappa - \varkappa)/\xi^2\big] =0.
	\end{equation}
	As $\xi \leq \hat\varkappa$, this implies that $\lim_{\mu\dto\mu_*}(\varkappa/\hat\varkappa) = 1$, as claimed.

	By Lemma~\ref{lem:step_1}, for any $0<\lambda_1<1<\lambda_2$, we have
	\begin{equation}
		\E\Mon\;\sim\!\sum_{k\in[\hat\varkappa\lambda_1,\,\hat\varkappa\lambda_2)}k\me^{-\mu k - u(k)}.
	\end{equation}
	At the same time,
	\begin{equation}
		\hat\varkappa\lambda_1\sum_{k\in[\hat\varkappa\lambda_1,\,\hat\varkappa\lambda_2)}\me^{-\mu k - u(k)}
		\,\leq\,
		\sum_{k\in[\hat\varkappa\lambda_1,\,\hat\varkappa\lambda_2)}k\me^{-\mu k - u(k)}
		\,<\,
		\hat\varkappa\lambda_2\sum_{k\in[\hat\varkappa\lambda_1,\,\hat\varkappa\lambda_2)}\me^{-\mu k - u(k)}.
	\end{equation}
	Because $\varkappa\sim\hat\varkappa$, Lemma~\ref{lem:step_1} also implies that these upper and lower bounds are asymptotically equivalent to $\varkappa\lambda_1S_\mu$ and $\varkappa\lambda_2S_\mu$ respectively. As $\lambda_1$ and $\lambda_2$ may be chosen arbitrarily close to 1, the assertion follows.
\end{proof}

\subsection{The incomplete gamma function regime}\label{sec:incGfun}
Let us now consider the critical regime, where $\lim_{x\to \infty} x^2 u^\pprime(x)\in \R$, see p.~\pageref{eq:GammaFun}. Using the definition of $\mu_*$ \eqref{eq:mu_star} and Lemma~\ref{lem: structure of u, const limit}, we can represent the function $u(\cdot)$ as
\begin{equation}\label{eq:uvthing}
	u(x)\;=\; - \mu_* \,x\,+\,(1-d) \ln x\,+\,v(x) \text{\quad with\quad}d\in \R,	\quad \lim_{x\to \infty} x v^\prime (x) =0.
\end{equation}
Notice that in this case the thermodynamic limit is achieved as $\mu\dto\mu_*\in\R$. To simplify notation, let us assume, without loss of generality, that $\mu_*=0$ and fix the scaling as $\varkappa=1/\mu$. If $d>1$, this scaling is asymptotically equivalent to that determined by equation \eqref{eq:step_scaling}. Notice that when $0\leq d \leq 1$, equation \eqref{eq:step_scaling} has no solutions for $\mu>\mu_*$, however this scaling is still appropriate as may be seen from the analysis below.
The scaled size distribution functions may now be defined as
\begin{equation}\label{eq:scaled_dist_fun_gFun}
	F_\mu(x;\vv\pn)\,\coloneqq\,\frac{1}{\mu\E\Mon}\;\sum_{\mu k\,\geq\,x}\pn_k.
\end{equation}
%
%
%
%
%

\begin{lemma}\label{lem:GfunL}
	Suppose that the condition \eqref{eq:uvthing} holds and $d\geq 0$.
	Then
	\begin{equation}
		\E\Mon \sim \mu^{-(d+1)} \me^{-v(1/\mu)}\Gfun(d+1),
		\qquad
		S_\mu(x/\mu)\;\sim\;\mu^{-d}\me^{-v(1/\mu)}\Gfun(x;d)
		\quad
		\text{ for all $x>0$}.
	\end{equation}
\end{lemma}
\begin{proof}
	We first compute
	\begin{equation} \label {eqn: S_mu; d geq 0}
		S_\mu(x/\mu)\,=\sum_{\mu k\,\geq\,x}k^{d-1}\me^{-\mu k-v(k)}
		\,=\,\mu^{-d}\me^{-v(1/\mu)}\sum_{\mu k\,\geq\,x}(\mu k)^{d-1}\me^{-\mu k}\me^{-v(k)+v(1/\mu)}\mu.
	\end{equation}
	Take any $\varepsilon>0$. As $\lim_{k\to \infty} k v^\prime (k) =0$, we have, $|v(x) - v(y)|\leq \varepsilon |\ln x - \ln y|$ for all large enough $x$ and $y$.
	Therefore, we have $|-v(k)+v(1/\mu)| \leq \varepsilon |\ln (\mu k)|$ for all small enough $\mu$.
	Sending $\mu\dto 0$ and then $\varepsilon \to 0$, we obtain that the last sum in equation \eqref{eqn: S_mu; d geq 0} converges to $\Gfun(x;d)$.

	The asymptotics of $\E \Mon$ follows similarly. The difference is that the summation \eqref{eq:sumMon} for $\E \Mon$ starts from $k=1$, and thus we must show that the contribution from the corresponding range of $k$ is negligible in the limit of vanishing $\mu$. Indeed, for any $\varepsilon\in (0,1)$, take $N=N(\varepsilon)$ such that $|v(x) - v(y)|\leq \varepsilon |\ln x - \ln y|$ for all $x,y\geq N$. Split the sum producing $\E \Mon$ into two parts,
	\begin{equation}\label{eq:split_sum}
		\sum_{k=1}^Nk^d\me^{-\mu k - v(k)}\;+\;\sum_{k=N+1}^\infty k^d\me^{-\mu k - v(k)}.
	\end{equation}
	As $v(k)\ll \ln k$, we have $\mu^{d+1}  \me^{v(1/\mu)} = \lO(1)$. The first sum in equation \eqref{eq:split_sum} my be bounded by \mbox{$N^d \exp\left\{\max_{k=1\ldots N} |v(k)|\right\}$}.
	Therefore
	\begin{equation}
		\lim_{\mu\to 0}\mu^{d+1}  \me^{v(1/\mu)} \E \Mon
		=
		\lim_{\varepsilon\to 0}
		\lim_{\mu\to 0}
		\mu^{d+1}  \me^{v(1/\mu)}  \sum_{k=N+1}^\infty  k^d \me^{-\mu k -v(k)}
		=
		\Gfun(d+1),
	\end{equation}
	concluding the proof.
\end{proof}

In the following theorem we prove that the limit shape is given by the incomplete gamma function if either $d>0$; or $d=0$ and $\lim_{x\to\infty} v(x) = -\infty$.
Notice that when $d=0$, we must require existence of the limit of $v(x)$ as $x\to\infty$, as it is not a consequence of our general assumptions that $\lim_{x\to\infty} x^2 v^\pprime(x) = 0$ and $\lim_{x\to\infty}xv^\prime(x)=0$. For example, consider $v(x)=\sin(\ln\ln x)$. (Alternatively, one could require sign-definiteness of $u^\pprime(x)$ or $v^\pprime(x)$ for all large enough $x$.)

%
\begin{theorem}\label{Thm: ln k + positive}
	Let $u(\cdot)$ and $v(\cdot)$ be related as in equation \eqref{eq:uvthing}, $F_\mu(x;\vv\pn)$ be as defined in equation \eqref{eq:scaled_dist_fun_gFun}. Assume either that $d>0$, or $d=0$  and $\lim_{x\to\infty} v(x) = -\infty$.
	Then for each $a>0$, $\epsilon>0$,
	\begin{equation}
		\lim_{\mu\dto0}\Prob \left\{\sup_{x\geq a}\big|F_\mu (x;\vv\pn) - F(x)| > \epsilon\right\}\,=\,0.
	\end{equation}
	Here $F(x) = \Gfun(x;d)/\Gfun(d+1)$; $\Gfun(x;d)$ is the incomplete gamma function, cf. \eqref{eq:GammaFun}.
\end{theorem}

\begin{proof}
	Notice that $F(x)$ is a decreasing function and $\lim_{x\to \infty}F(x) = 0$.
	By triangle inequalities, it is enough to show, for any fixed $x>0$,
	\begin{equation}
		\lim_{\mu\dto 0}\Prob_\mu \left\{|F_\mu (x;\vv\pn) - F(x)| > \epsilon\right\} = 0.
	\end{equation}
	Lemma \ref{lem:GfunL} shows that
	\begin{equation}
		\lim_{\mu\dto0}\E F_{\mu}(x;\vv \pn)
		\,=\,
		\lim_{\mu\dto0} \dfrac{1}{\mu\E\Mon} S_\mu(x/\mu)
		\,=\,
		F(x).
	\end{equation}
	It remains to show that
	\begin{equation}
		\lim_{\mu\dto 0}\Prob_\mu \left\{|F_\mu (x;\vv\pn) -  \E F_{\mu}(x;\vv \pn)|\,>\,\epsilon\right\} = 0.
	\end{equation}
	By Chebyshev's inequality and using that $\pn_k$-s are Poisson random variables,
	\begin{equation}\label{eq:Cheb_Var}
		\Prob_\mu \left\{|F_\mu (x;\vv\pn) -  \E F_{\mu}(x;\vv \pn)| > \epsilon\right\}
		\,\leq\,
		\epsilon^{-2}
		\V F_{\mu}(x;\vv \pn)
		\,=\,
		\dfrac{\epsilon^{-2}}{\mu\E \Mon}\E F_{\mu}(x;\vv \pn).
	\end{equation}
	By Lemma \ref{lem:GfunL}, $\mu\E \Mon \sim \mu^{-d} \me^{-v(1/\mu)}$, which tends to infinity as $\mu$ tends to zero.
\end{proof}

\subsection{Limit process regime}\label{sec:limitProcess}
Formula \eqref{eq:Cheb_Var} above and the accompanying argument show exactly why there is no limit shape if $d=0$ and $\lim_{x\to\infty}v(x)\in\R$: the variance of the scaled distribution functions does not vanish in the thermodynamic limit.  Instead of the limit shape, however, we get a limit process. This happens whenever the function $u(\cdot)$, see \eqref{eq:smallE}, may be represented as,
\begin{equation}\label{eq:C2c}
	u(x)\,=\,-\mu_*x + \ln x + v(x),
	\qquad\text{with}\qquad
	v(x)\to C\in\R
	\quad\text{as}\quad x\to\infty.
\end{equation}
As before, we may set $\mu_*=0$ and consider the limit $\mu\dto0$ rather than $\mu\dto\mu_*$. Computing the expected number of monomers in the system, we get,
\begin{equation}
	\E\Mon
	\,=\,
	\sum_{\mu k\,\geq\,0}
	\me^{-\mu k - v(k)}
	\,\sim\,\frac{1}{\mu}
	\int_0^\infty\me^{-t-C}\md t
	\,=\,\frac{\me^{-C}}{\mu}.
\end{equation}
Thus in this case, formula \eqref{eq:scaled_dist_fun_gFun} for  the scaled size distribution functions may be replaced with a simpler equivalent,
\begin{equation}\label{eq:scDFC3c}
	F_\mu(x;\vv\pn)\,\coloneqq\,\sum_{\mu k\,\geq\, x}\pn_k.
\end{equation}
The ``vertical'' scaling in this case is $\mu$-independent, as $\mu\E\Mon\to \me^{-C}$ in the limit.

As mentioned above, in this scaling the variance of $F_\mu(x;\vv\pn)$ does not vanish as $\mu\dto0$, thus we are expecting that the latter converges to a stochastic process rather than to a deterministic limit shape. Let us define a random measure on $\R^+$ corresponding to $F_\mu(x;\vv\pn)$ (its measure-valued derivative up to a minus sign):
\begin{equation}
	\pi_\mu(\md x;\vv\pn) \,\ass\,\sum_{k=1}^\infty \pn_k \delta(x-\mu k)\md x,
	\quad \text{i.e.,}\quad
	\pi_\mu(\set A;\vv\pn) \,=\,\sum_{\mu k \in\set A}^\infty \pn_k,\quad \set A\subset\R^+.
\end{equation}
The $\pn_k$-s are independent Poisson random variables with parameters $\me^{-\mu k - v(k)}/k$, consequently, for any finite sequence of disjoint intervals $[a_j,b_j)$, $j=1\ldots n$, the quantities $\pi_\mu([a_j,b_j);\vv\pn)$ are also independent Poisson random variables with parameters \smash{$\sum_{a_j\,\leq\,\mu k\,<\,b_j}\me^{-\mu k - v(k)}/k$}.
Computing
\begin{equation}\label{eq:PoiInt}
	\begin{split}
		\lim_{\mu\dto0} \sum_{a\,\leq\,\mu k\,<\,b}
		\dfrac1k\me^{-\mu k - v(k)}
		\,=\,
		\me^{-C}\int_a^b\frac{\me^{-x}}{x} \md x.
	\end{split}
\end{equation}
we conclude that $\pi_\mu$  converges in distribution to $\pi$, a Poisson point process on $\R^+$ with intensity given by \smash{$\me^{-C-x}/x$}, see \cite{kerstan1978infinitely}. Therefore we have just proved
\begin{proposition}\label{prop:process}
	Let $u(k)$ satisfy equation \eqref{eq:C2c} and the scaled size distribution function be defined as in equation \eqref{eq:scDFC3c}. Then its derivative, $\pi_\mu(x;\vv\pn)$,  converges in distribution to a Poisson point process with intensity \smash{$\me^{-C-x}/x$} as prescribed in equation \eqref{eq:PoiInt}. Alternatively,
	\begin{equation}
		F_\mu(x;\vv\pn) \,\stackrel{d}{\rightarrow}\,\pi([x,\infty)), \quad x\in \R^+.
	\end{equation}
\end{proposition}

\vspace{2ex}
Notice that the range of summation in formula \eqref{eq:scDFC3c} may be rewritten as $k\,\geq\,x\me^{C}\E\Mon$. If we remove the expectation and discard the factor of $\me^C$, we get an alternative (random) scaling for the size distribution function:
\begin{equation}
	\tilde F_\mu(x;\vv\pn)\,=\,\sum_{k\;\geq\;x\Mon(\vv\pn)}\pn_k.
\end{equation}
This scaling ensures that $\tilde F_\mu(x;\vv\pn)=0$ whenever $x>1$, and thus the size distribution function may be regarded as a random partition of the interval $[0,1]$ in the sense that \smash{$\tilde F_\mu(a;\vv\pn)-\tilde F_\mu(b;\vv\pn)$} is the number of subintervals with lengths in $[a,b)$ in such a partition. The limiting object in this case is a Poisson-Dirichlet distribution that describes cycle lengths in a random permutation \cite{kingman1975random,holst2001poisson}. This, of course, is not a coincidence, as the functional form of the function $u(\cdot)$, given in equation \eqref{eq:C2c}, implies that in this regime the resulting grand canonical Gibbs measure on partitions is exactly the Poissonization of the uniform (Dirichlet-Haar) measures on symmetric groups.

\subsection{No limit shape regimes}\label{sec:noshape}
Let us take a look at the regimes in which the limit shape does not exist. This may happen for several reasons: either because the number of monomers, $\E\Mon$, remains finite (as $\mu\dto\mu_*$); because the grand canonical measures are supported on non-summable sequences $\vv\pn$ which do not correspond to partitions of finite integers/sets; or if the limit shape technically exists, but is given by a trivial function identically equal to zero.

\paragraph{\mybull  $\boldsymbol{x^2 u^\pprime(x)\to -\infty}$ with $\boldsymbol{u^\prime(x) \dto -\infty}$.} We claim that the sum $\sum_{k=1}^\infty \pn_k$ is almost surely infinite for all $\mu$ (and thus so is $\Mon$). Indeed, as $\mu_*=\infty$, $S_\mu = \E \sum_{k=1}^\infty \pn_k = \infty$ for all $\mu\in \R$. The following proposition shows, that the grand canonical measures are concentrated on non-summable sequences, and thus there are no limit shapes in these scenarios:

\begin{proposition} \label {lem: concentration on infty config}
	Let $S_\mu = \infty$. Then
	\begin{equation}
		\Prob \left\{ \sum_{k=1}^\infty p_k =\infty \right\}\,=\,1,
		\text{\quad and consequently,\quad}
		\Prob \big\{\Mon=\infty \big\}\,=\,
		\Prob \left\{ \sum_{k=1}^\infty kp_k =\infty \right\}\,=\,1.
	\end{equation}
\end{proposition}
\begin{proof}
	Recall that $\pn_k$-s are independent Poisson random variables, thus their sum, $\sum_{k=1}^m\pn_k$, is also Poisson with parameter \smash{$A_m\ass\sum_{k=1}^m \me^{-\mu k - u(k)}$}, see Section \ref{ssec:lim_shapes}. For any fixed $N\in\N$ and all large enough $m$, we have
	\begin{equation}
		\Prob\left\{ \sum_{k=1}^\infty p_k \leq N  \right\}
		\,\leq\,
		\Prob\left\{ \sum_{k=1}^m p_k \leq N  \right\}
		\,=\,
		\me^{-A_m} \sum_{k=0}^N  \dfrac{A_m^k}{k!}
		\,\leq\,
		\me^{-A_m}A_m^N.
	\end{equation}
	Sending $m$ to infinity, we obtain, $\Prob\left\{ \sum_{k=1}^\infty \pn_k \leq N  \right\} =0$ for all $N\in\N$, yielding the assertion.
\end{proof}

\paragraph{\mybull  $\boldsymbol{x^2 u^\pprime (x)\to -\infty}$, $\boldsymbol{u^\prime(x) \to \const}$.}
Recall formula \eqref{eq:sumMon} for the expected number of monomers, $\E\Mon$. By an argument similar to that in Lemma \ref{lem: Assm2impliesAssm1}, we can get that $\mu_* = -\lim_{x\to\infty}u^\prime (x)$ and $\left( \mu_* x+u(x)\right) / \ln x \to \infty$. This implies that $\E\Mon<\infty$ when $\mu=\mu_*$, i.e., the expected number of monomers remains finite as $\mu\dto\mu_*$ and thus the proper thermodynamic limit cannot be achieved in the setting of grand canonical ensembles. 


\paragraph{\mybull $\boldsymbol{\lim_{x\to\infty} x^2 u^\pprime(x) \in \R}$.}
In this regime, $u(\cdot)$ has the form $u(x) = -\mu_* x + (1-d)\ln x + v(x)$, where $d\in \R$ and $\lim_{x\to \infty} x v^\prime (x) =0$, see equation \eqref{eq:uvthing}.
In Theorem~\ref{Thm: ln k + positive} and Proposition~\ref{prop:process} we considered cases when either $d > 0$ or $d=0$ with $\lim_{x\to\infty}v(x) \in \R \cup \{-\infty\}$. When $d<-1$,  $\E\Mon$ is dominated by a converging series, $\propto\sum k^d$, and so is finite for all $\mu$. Thus, similarly to the case above, there is no thermodynamic limit in this regime. It remains to discuss cases when $d\in[-1,0)$, or $d=0$ with $\lim_{x\to \infty} v(x) = \infty$.

%


As earlier, without loss of generality, we may set $\mu_*=0$. Consider any scaling, such that $\lim_{\mu\to 0} \varkappa(\mu) = \infty$, which is necessary to get a continuous limit for the size distribution functions. For any $y>x>0$, let
\begin{equation}\label{eq:noshapesum}
	F_\mu(x,y;\vv\pn)\,\ass\, \dfrac{\varkappa} {\E \Mon} \sum_{\varkappa x\,\leq\, k \,<\,\varkappa y}\pn_k.
\end{equation}
Then we have $F_\mu(x;\vv\pn) = F_\mu(x,y;\vv\pn) + F_\mu(y;\vv\pn)$, cf.~formula \eqref{eq:SDF_Ver}. We know that $\sum_{\varkappa x\,\leq\, k \,<\,\varkappa y}\pn_k$ is of Poisson distribution with parameter $S_\mu(\varkappa x, \varkappa y)$. Let us show that
\begin{equation} \label {eqn: S_mu xy vanishes}
	\lim_{\mu\dto 0}S_\mu(\varkappa x, \varkappa y)=0
\end{equation}
for all $x$ and $y$ with $0<x<y$. In fact, in the case $d<0$, as $v(x)\ll \ln x$, we have
\begin{equation}
	S_{\mu_*}
	\,=\,
	\sum_{k\geq 1} \me^{\,(d-1)\ln k - v(k)} \,<\, \infty.
\end{equation}
Notice that $S_\mu(\varkappa x, \varkappa y) \leq S_{\mu_*}(\varkappa x, \varkappa y)$.
As $\varkappa \to \infty$, equation \eqref{eqn: S_mu xy vanishes} follows, cf.\,\eqref{eq:main_sum}.
We now turn to the case when $d=0$ and $\lim_{x\to \infty} v(x) = \infty$.
Note that
\begin{equation}
	S_\mu(\varkappa x, \varkappa y)
	\,=\,
	\sum_{\varkappa x\,\leq\, k \,<\,\varkappa y} \dfrac1k \me^{-\mu k - v(k) }
	\,=\,
	\me^{-v(\varkappa)}\sum_{\varkappa x\,\leq\, k \,<\,\varkappa y} \dfrac1k \me^{-\mu k - (v(k) - v(\varkappa))}.
\end{equation}
Take any $\varepsilon>0$. Notice that $|v(k_1) - v(k_2)| \leq \varepsilon |\ln(k_1/k_2)|$ for $k_1,k_2$ large, cf.\,Lemma \ref{lem: structure of u, const limit}.
Then for $\varkappa$ large, it holds
\begin{equation}
	S_\mu(\varkappa x, \varkappa y)
	\,\leq\,
	c_{x,y,\varepsilon} \me^{-v(\varkappa)}  \sum_{\varkappa x\,\leq\, k \,<\,\varkappa y} \dfrac1k
	\,\leq\,
	c_{x,y,\varepsilon} \me^{-v(\varkappa)} \ln \left(\dfrac{y}{x}\right)
\end{equation}
where $c_{x,y,\varepsilon}$ is a constant depending only on $x$, $y$, and $\varepsilon$.
By the assumption that $\lim_{x\to\infty} v(x) = \infty$, we obtain \eqref{eqn: S_mu xy vanishes}.

As a consequence of \eqref{eqn: S_mu xy vanishes}, we have
\begin{equation}
	\lim_{\mu\dto0}\Prob \left\{ F_\mu(x,y;\vv\pn) = 0 \right\}
	\,=\,
	\lim_{\mu\dto0}\Prob \left\{  \sum_{\varkappa x\,\leq\, k \,<\,\varkappa y}\pn_k = 0 \right\}
	\,=\,
	\lim_{\mu\dto0} \me^{-S_{\mu}(\varkappa x, \varkappa y)}
	\,=\,1.
\end{equation}
Assume now that $F_\mu(x;\vv\pn)$ converges in distribution to some random variable $\xi$. Then $F_\mu(y;\vv\pn)$ must converge in distribution to $\xi$ as well. In other words, for all $x>0$, $F_\mu(x;\vv\pn)$ converges in distribution to $\xi$. However,
\begin{equation}
	\int_0^\infty \E \xi  \md x
	\,\leq\,
	\int_0^\infty \liminf_{\mu\dto 0} \E F_\mu(x;\vv\pn)  \md x
	\,\leq\,
	\liminf_{\mu\dto 0}
	\int_0^\infty \E F_\mu(x;\vv\pn)  \md x
	\,=\,
	1.
\end{equation}
This implies that $\xi\equiv0$, i.e, the only admissible limit shape is the zero function.

\section{Local profiles of the step function shape}\label{sec:locProf}
As discussed in Section~\ref{ssec:StepFun}, whenever $x^2u^\pprime(x)\to\infty$ as $x\to\infty$, the limit shape is given by the step function. In this section we investigate its local profile near the discontinuity point. We consider the shifted and rescaled size distribution function as defined in equation \eqref{eq:shift_shape}:
\begin{equation}
	G_\mu(x;\vv\pn)
	\,=\,
	\dfrac{\varkappa}{\E \Mon}\,\sum_{k\,\geq\,\zeta x + \varkappa}\pn_k,
	\quad x\,\in\,\R.
\end{equation}
The new parameter $\zeta = \zeta(\mu)$ controls the scaling in the vicinity of the discontinuity point.

\subsection{Gaussian regime} \label{sec: gaussian regime}
We start from the case when the local limit shape is given by a Gaussian integral. This happens when $u^\pprime(x)\to 0$ with additional assumption that $x^2 u^\pprime(x)$ is non-decreasing and $u^\pprime(x)$ is non-increasing.
New local limit shapes along subsequences of $\mu$ might show up if monotonicity is not assumed.
See Section \ref{monotonicity matters} for a more detailed discussion.

\begin{theorem} \label {Thm: Gaussian limit}
	Assume that $\displaystyle \lim_{x\to \infty}u^\pprime (x) = 0$, $\displaystyle \lim_{x\to \infty}x^2 u^\pprime (x) = \infty$ and both $u^\pprime(x)$ and $x^2 u^\pprime(x)$ are monotone.
	Let
	\begin{equation} \label{eqn: local shape fcn}
		G_\mu(x;\vv\pn)\,\coloneqq\,
		\dfrac{\varkappa} { \E \Mon}
		\sum_{k - \varkappa \geq x / \sqrt{u^\pprime (\varkappa)}  }\pn_k.
	\end{equation}
	Then,
	for each $\epsilon>0$, we have
	\begin{equation}
		\lim_{\mu\dto \mu_*}\Prob  \left\{\sup_{x\in \R}|G_\mu (x;\vv\pn) - G(x)| > \epsilon\right\} = 0
	\end{equation}
	where
	\begin{equation}
		G(x) \,=\,\frac{1}{\sqrt{2\pi}}\int_x^\infty  \me^{-t^2/2}\md t
	\end{equation}
\end{theorem}
\begin{proof}
	The theorem will follow if we show, for all $x\in \R$,
	\begin{equation}
		\lim_{\mu\dto \mu_*} \E G_\mu (x;\vv \pn)
		=
		\dfrac{1}{\sqrt{2\pi}}  \int_x^\infty  \me^{-t^2/2} \md t.
	\end{equation}
	Recall the notation in equation \eqref{eq:main_sum}. By Lemma \ref{lem: polymer and monomer number}, $\varkappa/\E \Mon\sim1/S_\mu$. Therefore
	\begin{equation}
		\E G_\mu (x;\vv \pn)
		\sim
		S_\mu \left(\varkappa + x / \sqrt{u^\pprime (\varkappa)} \right)/S_\mu.
	\end{equation}
	Thus it suffices to show that
	\begin{align}\label{eqn: shifted limit 2pi}
		\lim_{\varkappa\to\infty}
		\sqrt{u^\pprime (\varkappa)}
		\me^{\mu \varkappa +u(\varkappa)}
		S_\mu
		\, & =\,\int_{-\infty}^\infty  \me^{-t^2/2} \md t, \\\nonumber
		\lim_{\varkappa\to\infty}
		\sqrt{u^\pprime (\varkappa)}
		\me^{\mu \varkappa +u(\varkappa)}
		S_\mu \left(\varkappa + x / \sqrt{u^\pprime (\varkappa)} \right)
		   & =\,\int_x^\infty  \me^{-t^2/2} \md t.
	\end{align}
	We will prove the first limit only; the second follows by a similar argument.
	By Lemma \ref{lem:step_1}, for any $\varepsilon>0$, $S_\mu\sim S_\mu((1-\varepsilon)\varkappa, (1+\varepsilon)\varkappa)$. Thus
	\begin{equation} \label{eqn: S_mu Gaussian}
		\sqrt{u^\pprime (\varkappa)}
		\me^{u(\varkappa) + \mu \varkappa} S_\mu
		\sim
		\sum_{(1-\varepsilon) \varkappa \leq k\leq (1+\varepsilon) \varkappa}
		\me^{ - \mu (k-\varkappa) - \left(u(k) - u(\varkappa) \right) }
		\sqrt{u^\pprime (\varkappa)}
		=
		\sum_{-\varepsilon \varkappa \leq k\leq \varepsilon \varkappa}
		\me^{ - \mu k - \left(u(k+\varkappa) - u(\varkappa) \right) }
		\sqrt{u^\pprime (\varkappa)}.
	\end{equation}
	By the assumption, $\lim_{x\to \infty} x^2 u^\pprime(x) = \infty$, and thus $\varepsilon \varkappa \sqrt{u^\pprime(\varkappa)}\to \infty$. As $u^\pprime(\varkappa)\to 0$,
	\begin{equation}
		\lim_{\varepsilon\dto 0}
		\lim_{\mu\dto \mu_*}
		\sum_{-\varepsilon \varkappa \leq k\leq \varepsilon \varkappa}
		\me^{-u^\pprime(\varkappa) k^2/2}
		\sqrt{u^\pprime (\varkappa)}
		=
		\int_{-\infty}^{\infty} \me^{- x^2/2} \md x
		=
		\sqrt{2\pi}.
	\end{equation}
	We now show that the exponential function in the last summation of formula \eqref{eqn: S_mu Gaussian} may be replaced by \smash{$\me^{-u^\pprime(\varkappa) k^2/2}$}. Using Taylor expansion, we can find $\xi=\xi(k,\varkappa)$ in between $\varkappa$ and $k+\varkappa$, such that
	\begin{equation} \label{eqn: gaussian limit taylor}
		-\mu k - \left(u(k+\varkappa) - u(\varkappa) \right)
		=- \left[u(k+\varkappa) - u(\varkappa) - u^\prime (\varkappa) k \right]
		=-\dfrac{u^\pprime(\xi)}{2} k^2
		=-\dfrac{u^\pprime(\xi)}{u^\pprime(\varkappa)} \dfrac12 u^\pprime(\varkappa) k^2.
	\end{equation}
	As $u^\pprime(x)$ is decreasing,
	we must have $u^\pprime((1+\varepsilon) \varkappa)/u^\pprime(\varkappa) \leq u^\pprime(\xi)/u^\pprime(\varkappa)\leq 1$ for $ 0 \leq k \leq \varepsilon \varkappa$; and $1\leq u^\pprime(\xi)/u^\pprime(\varkappa)\leq u^\pprime((1-\varepsilon) \varkappa)/u^\pprime(\varkappa)$ for $-\varepsilon \varkappa \leq k < 0$.
	Therefore,
	\begin{align}
		\sum_{0 \leq k\leq \varepsilon \varkappa}
		\me^{ -  u^\pprime(\varkappa) k^2 /2}
		\; & \leq
		\sum_{0 \leq k\leq \varepsilon \varkappa}
		\me^{ - \mu k - \left(u(k+\varkappa) - u(\varkappa) \right) }
		\;\leq
		\sum_{0 \leq k\leq \varepsilon \varkappa}
		\me^{ - \frac{u^\pprime((1+\varepsilon) \varkappa)}{u^\pprime(\varkappa)}u^\pprime(\varkappa) k^2/2 }
		\nonumber \\
		\sum_{-\varepsilon \varkappa \leq k<0}
		\me^{ -  u^\pprime(\varkappa) k^2 /2}
		\; & \geq
		\sum_{-\varepsilon \varkappa \leq k<0}
		\me^{ - \mu k - \left(u(k+\varkappa) - u(\varkappa) \right) }
		\;\geq
		\sum_{-\varepsilon \varkappa \leq k< 0}
		\me^{ - \frac{u^\pprime((1-\varepsilon) \varkappa)}{u^\pprime(\varkappa)} u^\pprime(\varkappa) k^2 /2}
	\end{align}
	As $x^2 u^\pprime(x)$ is increasing, $(1-\varepsilon)^2u^\pprime((1-\varepsilon)\varkappa)\;\leq\;u^\pprime(\varkappa)\;\leq\;(1+\varepsilon)^2u^\pprime((1+\varepsilon)\varkappa)$. Thus we conclude that
	\begin{equation}
		\begin{split}
			&\lim_{\varepsilon\dto 0}
			\lim_{\mu\dto \mu_*}
			\sum_{0 \leq k\leq \varepsilon \varkappa}
			\me^{ - \mu k - \left(u(k+\varkappa) - u(\varkappa) \right) }
			\sqrt{u^\pprime (\varkappa_m)}
			=
			\int_{0}^{\infty} \me^{-x^2/2} \md x,\\
			&
			\lim_{\varepsilon\dto 0}
			\lim_{\mu\dto \mu_*}
			\sum_{-\varepsilon \varkappa \leq k < 0}
			\me^{ - \mu k - \left(u(k+\varkappa) - u(\varkappa) \right) }
			\sqrt{u^\pprime (\varkappa)}
			=
			\int_{-\infty}^0 \me^{-x^2/2} \md x.
		\end{split}
	\end{equation}
	And the limit in equation \eqref{eqn: shifted limit 2pi} follows.
\end{proof}

\subsection{Discrete Gaussian regime}
When $\lim_{x\to \infty}u^\pprime (x) = c$ for some $0<c<\infty$, for the shifted limit shape, no rescaling is needed in the ``horizontal'' direction and the shape is a discrete Gaussian.
\begin{theorem}
	Assume that $\lim_{x\to \infty}u^\pprime (x) = c$ for some $0<c<\infty$. Let
	\begin{equation}
		G_\mu(x;\vv\pn)\,\coloneqq\,
		\dfrac{\varkappa} { \E \Mon}
		\sum_{k - \varkappa \,\geq\,x}\pn_k.
	\end{equation}
	Then for each $\epsilon>0$
	\begin{equation}
		\lim_{\mu\dto \mu_*}\Prob  \left\{\sup_{x\in \R}|G_\mu (x;\vv\pn) - G(x)| > \epsilon\right\} = 0
	\end{equation}
	where $G(x) =\dfrac 1{M_c} \sum_{k\geq x} \me^{-c k^2/2}$ and $M_c = \sum_{k\in  \Z} \me^{-c k^2/2}$.
\end{theorem}
\begin{proof}
	Similar to proof of Theorem \ref{Thm: Gaussian limit}, it suffices to show that for all $x\in \R$
	\begin{equation} \label{eqn: discrete Gaussian limit}
		\E G_\mu(x;\vv\pn)
		=
		\dfrac{\varkappa} { \E \Mon}
		S_\mu(\varkappa + x)
		\to
		\dfrac 1{M_c} \sum_{k\,\geq\,x} \me^{-c k^2/2}
		\quad \text{as $\mu\dto \mu_*$}.
	\end{equation}
	Recall that $\mu = -u^\prime(\varkappa)$. Then for some $\xi=\xi(k,\varkappa)$ in between $\varkappa$ and $k$, we have
	\begin{equation}
		\me^{u(\varkappa) + \mu \varkappa}
		S_\mu(\varkappa + x)
		=
		\me^{u(\varkappa) + \mu \varkappa}
		\sum_{k - \varkappa \geq x} \me^{-\mu k - u(k)}
		=
		\sum_{k-\varkappa \geq x}
		\me^{-u^\pprime(\xi) (k-\varkappa)^2/2}.
	\end{equation}
	Notice that $\varkappa\upto \infty$ as $\mu\dto \mu_*$. Then, for $\mu$ close to $\mu_*$, we can relabel $k-\varkappa$ as $k$ so that
	\begin{equation}
		\sum_{k-\varkappa \geq x}
		\me^{-u^\pprime(\xi) (k-\varkappa)^2/2}
		=
		\sum_{k \geq x}
		\me^{-u^\pprime(\xi) k^2/2}
		\to
		\sum_{k \geq x}
		\me^{-c k^2/2}.
	\end{equation}
	In the last convergence we have used $\xi \geq \min\{\varkappa, \varkappa + x \}$ and $\lim_{k\to \infty}u^\pprime (k) = c$. It remains to consider the contribution from $\varkappa/\E \Mon$. By Lemmas \ref{lem:step_1} and \ref{lem: polymer and monomer number}, for any fixed $\varepsilon \in (0,1)$, we have
	\begin{equation} \label{eqn: discrete gaussian EM}
		\me^{u(\varkappa) + \mu \varkappa} \dfrac{\E\Mon}{\varkappa}
		\sim
		\me^{u(\varkappa) + \mu \varkappa}
		\sum_{k\geq (1-\varepsilon) \varkappa} \me^{-\mu k - u(k)}
		=
		\sum_{k \geq -\varepsilon\varkappa}
		\me^{-u^\pprime(\xi) k^2/2}.
	\end{equation}
	Here we have $\xi\geq (1-\varepsilon) \varkappa$. Notice that $u^\pprime(k)\to c>0$.
	Sending $\mu\dto \mu_*$, we obtain
	\begin{equation}
		\me^{u(\varkappa) + \mu \varkappa} \dfrac{\E\Mon}{\varkappa}
		\sim
		\sum_{k>- \infty}
		\me^{-c k^2/2}
		\eqqcolon
		M_c .
	\end{equation}
	Formula \eqref{eqn: discrete Gaussian limit} follows.
\end{proof}

\subsection{Hard step function regime}
In the case when $\lim_{x\to \infty}u^\pprime (x) = \infty$, the distribution of $p_k$ is so concentrated around $\varkappa$ that, for any scaling in $x$ direction, the local limit shape function is always the step function $\id_{(-\infty,0]}(x)$.
\begin{theorem}
	Assume that $\displaystyle \lim_{x\to \infty}u^\pprime (x) = \infty$.
	For any $\zeta = \zeta(\mu)\upto \infty$ as $\mu\dto\mu_*$, let
	\begin{equation}
		G_\mu(x;\vv\pn)\,\coloneqq\,
		\dfrac{\varkappa} { \E \Mon}
		\sum_{k - \varkappa \geq \zeta x}\pn_k.
	\end{equation}
	Then for each or each $\lambda_1<0$, $\lambda_2>0$, and $\epsilon>0$,
	\begin{equation}
		\lim_{\mu\dto \mu_*}\Prob \left\{\sup_{x\in \R \setminus (\lambda_1,\lambda_2)}|G_\mu (x;\vv\pn) - G(x)| > \epsilon\right\} = 0.
	\end{equation}
	where $G(x) = \id_{(-\infty,0]}(x)$.
\end{theorem}
\begin{proof}
	It suffices to show that $\lim_{\mu\dto \mu_*}\E G_\mu (x;\vv\pn) = G(x)$ for all $x\neq 0$. Following the proof of Theorem \ref{thn:step_1}, we show that \smash{$\sum_{k=1}^\infty \me^{-\mu k - u(k)}$}  concentrates on
	$k=\lfloor \varkappa \rfloor-1$, $\lfloor \varkappa \rfloor$, $\lfloor \varkappa \rfloor+1$, and $\lfloor \varkappa \rfloor+2$ where $\lfloor \varkappa \rfloor$ is the integer part of $\varkappa$.

	\smallskip
	{\bf \noindent Step 1.} Recall notation in equation \eqref{eq:main_sum}. Let us first show that
	$\lim_{\mu\dto \mu_*}S_\mu(\varkappa+3)/S_\mu\,=\,0$.
	Take any two integers $a,b$ such that $\lfloor\varkappa\rfloor+1<a<b$. Recall that $\mu = - u^\prime(\varkappa)$. Then for some $\xi$, $\xi^\prime$ with $a\leq \xi \leq b$ and $\varkappa  \leq \xi^\prime \leq \xi$,
	\begin{equation}
		\dfrac{\me^{-\mu b - u(b)}}{\me^{-\mu a - u(a)}}
		\,=\,
		\me^{ u^\prime(\varkappa) (b-a) - u^\prime (\xi) (b-a)}
		\,=\,
		\me^{-(\xi - \varkappa) (b-a) u^\pprime(\xi^\prime)}
		\,\leq\,
		\me^{-(a - \varkappa -1) (b-a) u^\pprime(\xi^\prime)}
		\,\leq\,
		\me^{-(b-a) u^\pprime(\xi^\prime)}.
	\end{equation}
	Here we used $\xi - \varkappa \geq a - (\lfloor\varkappa\rfloor+1) \geq 1$. Take some $A>0$. As $u^\pprime(x) \to \infty$, $\me^{-\mu b - u(b)}/\me^{-\mu a - u(a)} \leq \me^{-(b-a) A}$. Therefore for all large enough $\varkappa$-s, $S_\mu(\varkappa+3)/S_\mu\,\leq\,S_\mu(\varkappa+3)/S_\mu(2)\,\leq\,\me^{-A}$. Sending $A\to \infty$, we establish the desired limit.

	\smallskip
	{\bf \noindent Step 2.}
	It remains to show that $\lim_{\mu\dto \mu_*}S_\mu(1,\varkappa-1)/S_\mu\,=\,0$. By Lemma \ref{lem:step_1}, we may ignore first finitely many terms in the sum $S_\mu(1,\varkappa-1)$, i.e., for any fixed $N$\!, $\lim_{\mu\dto \mu_*}S_\mu(1,N)/S_\mu\,=\,0$.
	As $u^\pprime(x) \to \infty$, we may find $A(N)$ such that $\lim_{N\to\infty}A(N)=\infty$ and $u^\pprime(x) \geq A(N)$ for $x\geq N$. By a similar argument as in Step 1, $S_\mu(1,\varkappa-1)/S_\mu\,\leq\,S_\mu(1,\varkappa-1)/S_\mu(N+1,\varkappa)	\,\leq\,\me^{-A(N)}$. The proposition follows after sending $\mu\to \mu_*$ and then $N\to \infty$.
\end{proof}

\subsection{Remark on the monotonicity requirement} \label{monotonicity matters}
In section \ref{sec: gaussian regime} we required that $x^2 u^\pprime(x)$ be non-decreasing and $u^\pprime(x)$ be non-increasing to get the Gaussian limit for the local shape function. These monotonicity assumptions are essential in the sense that if either of them is not satisfied, either the Gaussian limit is not ensured or different limit shapes may occur along different subsequences of $\mu$. Let us construct an example illustrating this phenomenon.

First of all, using the scaling relation $u^\prime(\varkappa)=-\mu$, we can deduce that
\begin{equation}
	\int_{\varkappa}^{\varkappa+k} (\varkappa+k - x) u^\pprime (x)\md x
	\,=\,
	\mu k + u(k+\varkappa) - u(\varkappa).
\end{equation}
Define $u^\pprime(\cdot)$ in the following way: first let $u^\pprime(x) = 1$ for $0< x \leq 1$, $u^\pprime(x) = 1/2^n$ for $2^{n-1}< x \leq 2^n$, $n\in \N$.
(One could modify $u^\pprime$ in $[2^n, 2^n + 2^{-n}]$ by interpolating between the end-point values to obtain a continuous function, but it's not required.)
It is not hard to check that $x^2 u^\pprime(x)$ tends to infinity, however not in a monotone manner. Recall the local shape function $G_\mu(x;\vv\pn)$ as in \eqref{eqn: local shape fcn}. We will see that the resulting limit shape is only determined up to subsequences.

First, pick the sequence $\varkappa_n = 3\cdot 2^{n-1}$, $n\in\N$. The corresponding $\mu_n$ is then determined by the relation \eqref{eq:step_scaling}.
Assume $\varepsilon<1/3$. Note that $u^\pprime(x) = u^\pprime(\varkappa_n)$ for $(1-\varepsilon) \varkappa_n\,\leq\,x\,\leq\,(1+\varepsilon) \varkappa_n$.
Then, using $\mu_n = -u^\prime(\varkappa_n)$, we get,
\begin{equation}
	\begin{split}
		\sum_{\!\!\!-\varepsilon \varkappa_n\,\leq\,k\,\leq\,\varepsilon \varkappa_n}
		\exp\left\{ - \mu_n k - \big[u(k+\varkappa_n) - u(\varkappa_n)\big]\right\}
		\,=\,
		\sum_{\!\!\!-\varepsilon \varkappa_n\,\leq\,k\,\leq\,\varepsilon \varkappa_n}
		\exp\left\{-\-\dfrac{k^2}{2} u^\pprime(\varkappa_n)\right\}
	\end{split}
\end{equation}
Repeating the arguments of Theorem~\ref{Thm: Gaussian limit}, we obtain the standard Gaussian $G(x) \,=\,\int_x^\infty  \me^{-t^2/2}\md t/\sqrt{2\pi}$ as the limit of $G_\mu(x;\vv\pn)$.

We now show that a different limit shape may appear if we follow another sequence of $\varkappa$. Let $\tilde\varkappa_n = 2^n$ and, accordingly, $\tilde \mu_n = -u^\prime(\tilde\varkappa_n)$.
Assume $\varepsilon<1/2$. Now $u^\pprime(x)$ behaves differently on the different sides of $\varkappa_n$:  $u^\pprime(x) = u^\pprime(\varkappa_n)$ for $(1-\varepsilon) \varkappa_n\,\leq\,k\,\leq\, \varkappa_n$ and $u^\pprime(x) = u^\pprime(\varkappa_n)/2$ for $ \varkappa_n\,<\,k\,\leq\,(1+\varepsilon) \varkappa_n$:
\begin{align}
	\sum_{\!\!\!-\varepsilon \varkappa_n\,\leq\,k\,\leq0}
	\exp\left\{ - \mu_n k - \big[u(k+\varkappa_n) - u(\varkappa_n)\big]\right\}
	\; & =
	\sum_{\!\!\!-\varepsilon \varkappa_n\,\leq\,k\,\leq0}
	\exp\left\{-\-\dfrac{k^2}{2} u^\pprime(\varkappa_n)\right\};\nonumber \\
	\sum_{\!\!\!0\,<\,k\,\leq\,\varepsilon \varkappa_n}
	\exp\left\{ - \mu_n k - \big[u(k+\varkappa_n) - u(\varkappa_n)\big]\right\}
	\; & =
	\sum_{\!\!\!0\,<\,k\,\leq\,\varepsilon \varkappa_n}
	\exp\left\{-\-\dfrac{k^2}{4} u^\pprime(\varkappa_n)\right\}.
\end{align}
Repeating the arguments of Theorem~\ref{Thm: Gaussian limit} again, we now find that $G_\mu(x;\vv\pn)$ converges to
\begin{equation}
	\tilde G(x) \,=\, \frac{\sqrt 2}{(1+\sqrt2)\sqrt\pi} \int_x^\infty h(t)\md t,
\end{equation}
where $h(t) \,=\,\me^{-t^2/4}$ for $t> 0$ and
$h(t) \,=\,\me^{-t^2/2}$ for $t\leq 0$.

Notice, however, that a sequence of $\varkappa_n$ realizing the Gaussian limit as in Theorem~\ref{Thm: Gaussian limit} always exists if we require monotonicity of $u^\pprime(x)$.
Indeed, by Lemma \ref{lem: limit of ration}, we may find positive sequences $\varepsilon_n\downarrow 0$ and $\varkappa_m\uparrow \infty$, such that
\begin{equation} \label {eqn: estimates for ratio of u pprime}
	\lim_{m\to \infty} \varepsilon_m \varkappa_m\sqrt{u^\pprime(\varkappa_m)} =\infty,
	\quad
	\limsup_{m \to \infty}\dfrac{u^\pprime((1-\varepsilon_m) \varkappa_m)}{u^\pprime(\varkappa_m)} = 1,
	\quad
	\lim_{n \to \infty}
	\liminf_{m \to \infty} \dfrac{u^\pprime((1+\varepsilon_n) \varkappa_m)}{u^\pprime(\varkappa_m)} = 1.
\end{equation}
Thus we have $u^\pprime((1+\varepsilon_n) \varkappa_m)/u^\pprime(\varkappa_m)\leq u^\pprime(\xi)/u^\pprime(\varkappa_m)\leq 1$ for all $\xi\in [\varkappa_m, (1+\varepsilon_n) \varkappa_m)]$; and also \mbox{$1\leq u^\pprime(\xi)/u^\pprime(\varkappa_m)	\leq u^\pprime((1-\varepsilon_m) \varkappa_m)/u^\pprime(\varkappa_m)$} for all $\xi\in [(1-\varepsilon_m) \varkappa_m, \varkappa_m)$.
Therefore, using Taylor expansion, cf.\,\eqref{eqn: gaussian limit taylor},
\begin{align}
	\sum_{0 \leq k\leq \varepsilon_n \varkappa_m}
	\me^{ -u^\pprime(\varkappa_m) k^2/2 }
	\; & \leq
	\sum_{0 \leq k\leq \varepsilon_n \varkappa_m}
	\me^{ - \mu k - \left(u(k+\varkappa_m) - u(\varkappa_m) \right) }
	\;\leq
	\sum_{0 \leq k\leq \varepsilon_n \varkappa_m}
	\me^{ - \frac{u^\pprime((1+\varepsilon_n) \varkappa_m)}{u^\pprime(\varkappa_m)}u^\pprime(\varkappa_m) k^2/2 };
	\nonumber \\
	\sum_{-\varepsilon_n \varkappa_m \leq k<0}
	\me^{ - u^\pprime(\varkappa_m) k^2 /2}
	\; & \geq
	\sum_{-\varepsilon_n \varkappa_m \leq k<0}
	\me^{ - \mu k - \left(u(k+\varkappa_m) - u(\varkappa_m) \right) }
	\;\geq
	\sum_{-\varepsilon_m \varkappa_m\leq k< 0}
	\me^{ - \frac{u^\pprime((1-\varepsilon_m) \varkappa_m)}{u^\pprime(\varkappa_m)}u^\pprime(\varkappa_m) k^2/2 }.
\end{align}
Taking $m\to \infty$ and then $n\to\infty$, we obtain the Gaussian limit by \eqref{eqn: estimates for ratio of u pprime}.

\appendix

\appendix

\section{Appendix}
The following lemmas establish a few asymptotic properties of functions satisfying Assumption~\ref{Standing Assump}:

\begin{lemma} \label{lem: structure of u, const limit}
	Assume that $u(\cdot)$ is a twice-differentiable function satisfying
	\begin{equation}\label{eq:lim_fin_App}
		\lim_{x\to \infty} x^2 u^\pprime(x)\,=\,\gamma\,\in\,\R.
	\end{equation}
	Then there exist $c\in\R$ and $v(x)$ satisfying $\lim_{x\to \infty} x v^\prime (x) =0$, such that
	\begin{equation}
		\lim_{x\to \infty} u^\prime(x) = c,
		\quad
		u(x)\;=\;c\,x\,-\,\gamma \ln x\,+\,v(x).
	\end{equation}
	Notice that we also have, $v(x) \ll \ln x$ as $x\to\infty$.
\end{lemma}

\begin{proof}
	Equation \eqref{eq:lim_fin_App} implies that for any $\varepsilon>0$ there exists $x_0$, such that $\gamma -\varepsilon<x^2 u^\pprime(x)<\gamma+\varepsilon$ for all $x>x_0$. Dividing these inequalities by \smash{$x^2$} and integrating, we conclude that
	\begin{equation}
		| u^\prime (x_1) - u^\prime (x_2)| \,\leq\, (|\gamma|+\varepsilon) \left|1/{x_1} - 1/{x_2}\right|
	\end{equation}
	for all \mbox{$x_1, x_2>x_0$}. Thus $\lim_{x\to \infty} u^\prime(x)$ exists and is finite. Let us denote this limit by $c$ and define
	\begin{equation}
		v(x) \,\ass\,u(x) - c x + \gamma \ln x.
	\end{equation}
	%
	%
	%
	Directly from this definition we can calculate that $\lim_{x\to \infty} x^2 v^\pprime(x) =\lim_{x\to \infty} v^\prime(x) = 0$. Then for any $\varepsilon>0$ and $x,x_0$ large enough, we also have  $| v^\prime (x) - v^\prime (x_0)| \leq \varepsilon \left|1/x - 1/x_0\right| $. Sending $x_0$ to infinity, we obtain $-\varepsilon/x \leq v^\prime (x) \leq \varepsilon / x $, from which we conclude that $\lim_{x\to\infty} xv^\prime(x) = 0$ and $v(x) \ll \ln x$.
\end{proof}

\begin{lemma}\label{lem: Assm2impliesAssm1}
	Assume that $u(\cdot)$ is a twice-differentiable function satisfying
	\begin{equation}\label{eq:lem_xsquare}
		\lim_{x\to \infty} x^2 u^\pprime(x) = \infty.
	\end{equation}
	Then the following holds:
	\begin{enumerate}
		\item $\displaystyle\mu_*\,\ass\,-\lim_{x\to\infty}\frac{u(x)}{x}\,=\,-\lim_{x\to \infty} u^\prime (x)$.
		\item
		      If $\mu_* = -\infty$, then  $u(x)\gg x$. If $\mu_*\in\R$, then $\ln x \ll u(x) + \mu_* x \ll x$ and $\displaystyle\lim_{x\to \infty}[ u(x) + \mu_* x] = - \infty$.
		\item $\displaystyle\lim_{\mu \dto \mu_*} S_\mu = \infty$\quad (refer to p.\pageref{eq:main_sum} for the definition of $S_\mu$).
	\end{enumerate}
\end{lemma}

\begin{proof}
	{\it 1.}\quad The assumption \eqref{eq:lem_xsquare} implies that $u^\pprime(x) > 0$ for all large enough $x$.
	Thus both limits $\lim_{x\to \infty} u^\prime (x)$ and $\lim_{x\to \infty} u (x)$ exist. Observe that if the latter is finite, we must also have $\lim_{x\to \infty} u^\prime (x)=0$, which implies $\mu_*\,=\,-\lim_{x\to \infty} u^\prime (x)$. If $\lim_{x\to \infty} u (x)=\pm\infty$, $\mu_*\,=\,-\lim_{x\to \infty} u^\prime (x)$ as well by L'Hospital's rule. Note that this assertion also holds by similar argument if the limit in equation \eqref{eq:lem_xsquare} were negative infinity.

	\noindent{\it 2.}\quad If $\mu_* = -\infty$, i.e., $u^\prime(x)\to\infty$ as $x\to\infty$, $u(x)$ must be superlinear, i.e., $u(x)\gg x$. If $\mu_*\in\R$, $u(x)+\mu_*x$ must be sublinear. At the same time, property \eqref{eq:lem_xsquare} implies that for any fixed $M>0$ and all sufficiently large $x$, $u^\pprime(x) \geq M/x^2$. Integrating from $x$ to infinity, we get $u^\prime(x)\,+\,\mu_*\,\leq\,-\,M/x$.
	Another integration yields that $\lim_{x\to \infty} [u(x)+\mu_* x]\leq -M\ln x + C \to - \infty$ as $x\to\infty$.


	\noindent{\it 3.}\quad If $\mu_* = -\infty$, $S_\mu\to\infty$, because every term in the series \eqref{eq:main_sum} diverges as $\mu\dto\mu_*$. If $\mu_*\in\R$, the monotone convergence implies that $S_\mu\to S_{\mu_*}=\infty$. The latter series diverges, as its terms grow unboundedly by assertion (2).
\end{proof}

\vspace{3ex} The following lemma is used in Section~\ref{monotonicity matters} to demonstrate existence of non-Gaussian local shapes when the monotonicity assumption is omitted:

\begin{lemma} \label {lem: limit of ration}
	Consider a non-increasing function, $f: (0,\infty) \mapsto (0,\infty)$. Assume $\lim_{x\to \infty} x^2f(x) = \infty$.
	Then there exist sequences $\varepsilon_n \dto 0$ and $x_k \upto \infty$, such that
	\begin{equation} \label {eqn: limit ratio}
		\lim_{k\to \infty} \varepsilon_k x_k\sqrt{f(x_k)} =\infty,
		\quad
		\limsup_{k\to \infty} \dfrac{f(x_k - \varepsilon_k  x_k)}{f(x_k)}= 1,
		\quad
		\lim_{n \to \infty}
		\liminf_{k\to \infty} \dfrac{f(x_k + \varepsilon_n  x_k)}{f(x_k)}  = 1.
	\end{equation}
\end{lemma}
\begin{proof}
	Fix any sequence $\tau_n \dto 0$ and assume $\tau_n<1$ for all $n$.
	Let $\alpha_n$ to be a sequence such that  $0<(1+\tau_n)^2 \alpha_n <1$ for all $n$ and $\lim_{n\to \infty} \alpha_n =1$. Since $x^2 f(x)\to \infty$, we can find a sequence $z_k$  such that \smash{$\tau_k x\sqrt{f(x)}\geq k$} whenever $x\geq z_k$.
	We first show that there exist $y_k\upto\infty$, such that
	\begin{equation} \label {eqn: sequence y_k}
		y_k\geq z_k,
		\quad
		f(y_k + \tau_n  y_k)\,/f(y_k) \;\geq\; \alpha_n,
		\text{ for all } n\leq k.
	\end{equation}
	We first show that there exists $y_1\geq z_1$ such that
	\begin{equation}
		f(y_1 + \tau_1  y_1)\,/f(y_1) \;\geq \;\alpha_1.
	\end{equation}
	Suppose it does not hold, i.e.\,for all $x\geq z_1$ we have
	\begin{equation}
		f(x + \tau_1  x)  \;<\; \alpha_1 f(x).
	\end{equation}
	Let $u_k = (1+\varepsilon_1)^k z_1 $. Then $f(u_k) \leq {\alpha_1}^k f(z_1)$.
	Therefore, as $0<(1+\tau_1)^2 \alpha_1<1$
	\begin{equation}
		\lim_{k\to \infty} u_k^2 f(u_k)
		\;\leq\;
		\lim_{k\to \infty} \left[(1+\tau_1)^2 \alpha_1\right]^k z_1^2 f(z_1)
		\;=\;
		0
	\end{equation}
	which contradicts the assumption that $\lim_{x\to \infty} x^2 f(x) = \infty$.

	We now show that we may find $y_2\geq z_2$ such that
	\begin{equation}
		f(y_2 + \tau_1  y_2)\,/f(y_2) \;\geq\; \alpha_1, \quad
		f(y_2 + \tau_2  y_2)\,/f(y_2)\; \geq\; \alpha_2.
	\end{equation}
	Assume the opposite. Then for all $x\geq z_2$,
	\begin{equation}
		f(x + \tau_1  x)  < \alpha_1 f(x)
		\quad \text{or}\quad
		f(x + \tau_2  x)< \alpha_2 f(x) .
	\end{equation}
	Let $u_0 = z_2$ and $u_k = (1+  \tau_1)^{m_k} (1+  \tau_2)^{n_k} u_{k-1}$ where $m_k + n_k=1$, $m_k = 1$ if $f(u_{k-1} + \tau_1  u_{k-1}) < \alpha_1f(u_{k-1})$, $m_k=0$ otherwise. When $n_k = 1$, we have $f(u_{k-1} + \tau_2  u_{k-1}) < \alpha_2f(u_{k-1})$. Therefore,
	\begin{equation}
		f(u_k)
		\;<\;
		\alpha_1^{\sum_{j=1}^k m_j} \alpha_2^{\sum_{j=1}^k n_j} f(u_0),
	\end{equation}
	and
	\begin{equation}
		u_k^2 f(u_k)
		\;<\;
		\left[(1+\tau_1)^2 \alpha_1\right]^{\sum_{j=1}^k m_j}
		\left[(1+\tau_2)^2 \alpha_2\right]^{\sum_{j=1}^k n_j}
		u_0^2 f(u_0).
	\end{equation}
	As \smash{$\sum_{j=1}^k m_j + \sum_{j=1}^k n_j = k$},  at least one of the two sums goes to infinity as $k\to \infty$. Therefore we get
	$\lim_{k\to \infty} u_k^2 f(u_k) = 0$ which contradicts the assumption that $\lim_{k\to \infty} x^2 f(x) = \infty$.
	Similarly, for any $k\geq 3$ we can find $y_k$ such that
	\eqref{eqn: sequence y_k} holds.

	Define $\varepsilon_n\coloneqq \tau_n / 4$ and $x_k\coloneqq (1+\tau_k / 2)y_k$. We now show that \eqref{eqn: sequence y_k} implies the lemma. As $x_k\geq y_k\geq z_k$, it follows that \smash{$\varepsilon_k x_k\sqrt{f(x_k)} \to \infty$}.
	Note that, as $\tau_k<1$,
	\begin{equation}
		y_k
		= \left(1 - \dfrac{\tau_k / 2}{ 1 + \tau_k / 2}\right) x_k
		\leq \left(1 - \dfrac{\tau_k}{4}\right) x_k
		= (1 - \varepsilon_k) x_k.
	\end{equation}
	Therefore
	\begin{equation}
		\limsup_{k\to \infty} \dfrac{f(x_k - \varepsilon_k  x_k)}{f(x_k)}
		\;\leq\;
		\limsup_{k\to \infty} \dfrac{f(y_k)}{f(x_k)}
		\;\leq\;
		\limsup_{k\to \infty} \dfrac{f(y_k)}{f(y_k+\tau_k y_k)} =1.
	\end{equation}
	Finally, notice that for any fixed $n$ and all large enough $k$, $(1+\tau_n / 4)(1+\tau_k / 2) \leq 1+ \tau_n$. Thus
	\begin{equation}
		\liminf_{k\to \infty} \dfrac{f(x_k + \varepsilon_n  x_k)}{f(x_k)}
		\;=\;
		\liminf_{k\to \infty} \dfrac{f\left((1+\varepsilon_n) (1+\tau_k/2) y_k\right)}{f((1+\tau_k/2) y_k)}
		\;\geq\;
		\liminf_{k\to \infty} \dfrac{f\left((1+\tau_n) y_k\right)}{f(y_k)}
		\;\geq\; \alpha_n.
	\end{equation}
	Take $n\to \infty$ and note that $\lim_{n\to \infty} \alpha_n = 1$ to complete the proof.
\end{proof}

\bibliographystyle{plain}
\bibliography{../bibliography/bibl}

\end{document}